%% file: main.tex
\documentclass[acmsmall,nonacm]{acmart}
\input{packages}

\makeatletter
\def\@ACM@checkaffil{
    \if@ACM@instpresent\else
    \ClassWarningNoLine{\@classname}{No institution present for an affiliation}%
    \fi
    \if@ACM@citypresent\else
    \ClassWarningNoLine{\@classname}{No city present for an affiliation}%
    \fi
    \if@ACM@countrypresent\else
        \ClassWarningNoLine{\@classname}{No country present for an affiliation}%
    \fi
}
\makeatother

\acmYear{2024}\copyrightyear{2024}
\setcopyright{rightsretained}
\acmConference[PODC '24]{ACM Symposium on Principles of Distributed Computing}{June 17--21, 2024}{Nantes, France}
\acmBooktitle{ACM Symposium on Principles of Distributed Computing (PODC '24), June 17--21, 2024, Nantes, France}
\acmDOI{10.1145/3662158.3662810 }
\acmISBN{979-8-4007-0668-4/ 24/06}

\author{Diogo Avelãs}
\authornote{The two first authors contributed equally to this research.}
\affiliation{
  \institution{LASIGE, Faculdade de Ciências, Universidade de Lisboa}
  \city{}
  \country{Portugal}
  }
\email{dinoroba@proton.me}
\orcid{https://orcid.org/0009-0004-5838-1667}

\author{Hasan Heydari}
\authornotemark[1]
\affiliation{
  \institution{LASIGE, Faculdade de Ciências, Universidade de Lisboa}
  \city{}
  \country{Portugal}
  }
\email{hheydari@ciencias.ulisboa.pt}
\orcid{https://orcid.org/0000-0003-2309-2457}

\author{Eduardo Alchieri}
\affiliation{
  \institution{Universidade de Brasilia}
  \city{}
  \country{Brasil}
  }
\email{alchieri@unb.br}
\orcid{https://orcid.org/0000-0002-6022-3631}

\author{Tobias Distler}
\affiliation{
  \institution{Friedrich-Alexander-Universität Erlangen-Nürnberg}
  \city{}
  \country{Germany}
  }
\email{distler@cs.fau.de}
\orcid{https://orcid.org/0000-0002-2440-5366}

\author{Alysson Bessani}
\affiliation{
  \institution{LASIGE, Faculdade de Ciências, Universidade de Lisboa}
  \city{}
  \country{Portugal}
  }
\email{anbessani@ciencias.ulisboa.pt}
\orcid{https://orcid.org/0000-0002-8386-1628}

\renewcommand{\shortauthors}{Avelãs et al.}

\begin{document}

\ifthenelse{\boolean{extendedVersion}}{
    \title{Probabilistic Byzantine Fault Tolerance (Extended Version)}
}{
    \title{Probabilistic Byzantine Fault Tolerance}
}

\begin{abstract}
Consensus is a fundamental building block for constructing reliable and fault-tolerant distributed services. 
Many Byzantine fault-tolerant consensus protocols designed for partially synchronous systems adopt a pessimistic approach when dealing with adversaries, ensuring safety even under the worst-case scenarios that adversaries can create.
Following this approach typically results in either an increase in the message complexity (e.g., PBFT) or an increase in the number of communication steps (e.g., HotStuff).
In practice, however, adversaries are not as powerful as the ones assumed by these protocols.
Furthermore, it might suffice to ensure safety and liveness properties with high probability.
To accommodate more realistic and optimistic adversaries and improve the scalability of BFT consensus, we propose \probft (Probabilistic Byzantine Fault Tolerance).
\probft is a leader-based probabilistic consensus protocol with a message complexity of $O(n\sqrt{n})$ and an optimal number of communication steps that tolerates Byzantine faults in permissioned partially synchronous systems.
It is built on top of well-known primitives, such as probabilistic Byzantine quorums and verifiable random functions.
\probft guarantees safety and liveness with high probability even with faulty leaders, as long as a supermajority of replicas is correct and using only a fraction (e.g., $20\%$) of messages exchanged in PBFT.
We provide a detailed description of \probft's protocol and its analysis.
\end{abstract}



\maketitle

\section{Introduction}

\paragraph{Context}
Consensus is a fundamental building block for constructing reliable and fault-tolerant distributed services, where participants agree on a common value from the initially proposed values.
This problem is primarily used to implement state machine replication~\cite{schneider_1990,modsmart,distler21byzantine} and atomic broadcast~\cite{luan1990fault,modular_broadcast,pedone1998}, and attracted considerable attention in the last few years, mainly due to its significant role in blockchains~\cite{vukolic_2015,nakamoto_2008,Bes20} and decentralized payment systems (e.g.,~\cite{lokhava_2019}).
Due to its importance and widespread applicability, consensus has been extensively studied in diverse system models, considering various synchrony assumptions and a spectrum of failure models, ranging from fail-stop to Byzantine, across permissioned and permissionless settings~\cite{pass_2017,li2023quorum,cachin_2022,vassantlal_2023}.

Many Byzantine fault-tolerant (BFT) consensus protocols (e.g., PBFT~\cite{pbft} and HotStuff~\cite{yin19hotstuff}\footnote{Strictly speaking, PBFT~\cite{pbft} and HotStuff~\cite{yin19hotstuff} are state machine replication protocols, not consensus protocols.
In this paper, when we refer to PBFT and HotStuff, we specifically discuss the single-shot versions of these protocols presented in~\cite{makingByzConsLive}, which address the consensus problem.}) adopt a pessimistic approach when dealing with Byzantine participants and adversaries, ensuring the safety of protocols even when Byzantine participants behave completely arbitrarily under the worst-case scenarios that corruption and scheduling adversaries can create.
That is, they typically consider adversaries that choose a corruption strategy and manipulate the delivery time of messages based on the entire history of the system, including the past and current states of replicas, as well as their exchanged messages.

Protocols following the pessimistic approach are built on the idea of 
\begin{enumerate*}[label=(\arabic*)]
\item making non-revocable decisions by considering the opinions of a quorum of replicas, as opposed to relying solely on a single replica, and
\item requiring the quorums used for making decisions to intersect in at least one correct replica.
\end{enumerate*}
Although effective, ensuring deterministic quorum overlaps poses inherent challenges in achieving both resource efficiency and high performance.
Some protocols~(e.g.,~PBFT) approach this conflict by opting for low latency and applying message-exchange patterns with quadratic message complexity. 
However, this can be prohibitively expensive, especially for BFT systems with a large number of replicas. 
Other protocols~(e.g.,~HotStuff) aim at reducing message complexity at the cost of adding extra communication steps. 
Unfortunately, this approach leads to increased end-to-end response times.

In practice, however, implementing punishment mechanisms, like those used in existing Proof-of-Stake blockchains (e.g.,~\cite{eth,cason2021design}), can make it costly for Byzantine participants to deviate from its specification, knowing that their actions may lead to detection and subsequent punishment.
Besides, adversaries are not as powerful as the ones assumed by protocols following the pessimistic approach.
Accordingly, in many real-world applications, it is sufficient to assume a static corruption adversary, which chooses a corruption strategy at the beginning of the execution of a consensus instance, as well as an adversarial scheduling that manipulates the delivery time of messages independently of the sender's identifier and if it is faulty or not.
Given this, ensuring safety and liveness with a high probability might be acceptable in many practical scenarios. 
This paper describes a new protocol for these less pessimistic practical scenarios that requires less message exchanges and still keeps optimal latency.
This protocol is called \probft (Probabilistic Byzantine Fault Tolerance).  

\paragraph{Overview of \probft}
\probft is a BFT leader-based consensus protocol that operates in permissioned partially synchronous systems and probabilistically ensures liveness and safety properties.
It achieves the optimal \emph{good-case latency} of three communication steps~\cite{goodcaselatency}, just like PBFT, albeit with a message complexity of $O(n\sqrt{n})$.
Figure~\ref{fig:4replicas:pbft:probft:hotsfuff} compares \probft, PBFT, and HotStuff in terms of the number of communication steps and exchanged messages for different numbers of replicas.
\probft's resource efficiency and scalability improvements are enabled by a unique combination of building blocks that are usually not employed in traditional BFT protocols, including probabilistic quorum systems~\cite{malkhi01probabilistic}, a mechanism to configure the degree of communication redundancy, and a verifiable random function (\vrf)~\cite{verifiable_random}.

\ifthenelse{\boolean{extendedVersion}}{
\begin{figure}[t!]
    \begin{subfigure}[t]{0.5\textwidth}
        \centering
        \includegraphics[scale=0.9]{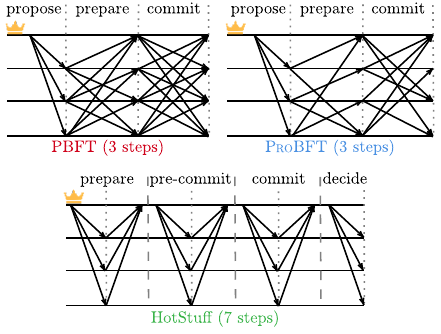}
        \caption{Message pattern and number of communication steps.}
        \label{fig:4replica:pbft:probft:hotstuff}
     \end{subfigure}
     \hfill
     \begin{subfigure}[t]{0.49\textwidth}
        \centering
    	\includegraphics[scale=0.65]{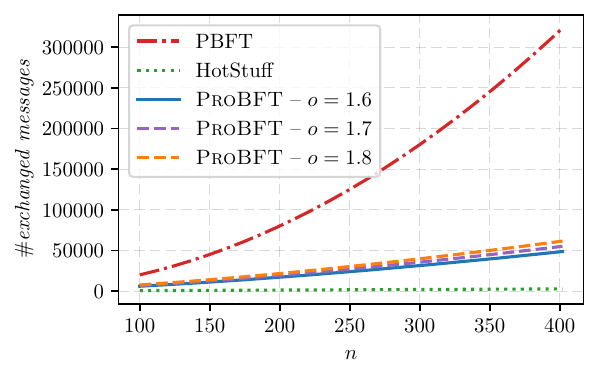}
        \caption{Number of exchanged messages.}
        \label{fig:pbft:probft:hotstuff}
    \end{subfigure}
    \caption{Comparing the normal case of three consensus protocols \--- PBFT, \probft, and HotStuff \--- regarding the number of communication steps and message complexity.}
    \Description[]{}
	\label{fig:4replicas:pbft:probft:hotsfuff}
\end{figure}
}{
\begin{figure}[!t]
    \begin{minipage}[b]{0.4\textwidth}
        \centering
        \adjustbox{left=4.2cm, right=0cm}{\includegraphics[scale=1.15]{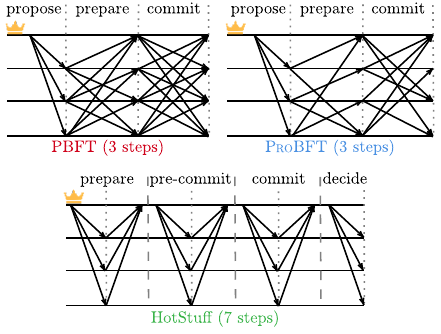}}
        \subcaption{Message pattern and number of communication steps.}
        \label{fig:4replica:pbft:probft:hotstuff}
     \end{minipage}
     \hfill
     \begin{minipage}[b]{0.4\textwidth}
        \centering
    	\adjustbox{left=4.2cm, right=0cm}{\includegraphics[scale=0.8]{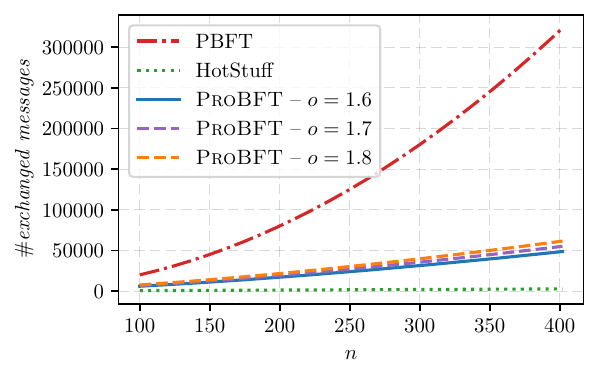}}
        \subcaption{Number of exchanged messages.}
        \label{fig:pbft:probft:hotstuff}
    \end{minipage}
    \caption{Comparing the normal case of three consensus protocols \--- PBFT, \probft, and HotStuff \--- regarding the number of communication steps and message complexity.}
    \Description[]{Comparing the normal case of three consensus protocols \--- PBFT, \probft, and HotStuff \--- regarding the number of communication steps and message complexity.}
	\label{fig:4replicas:pbft:probft:hotsfuff}
\end{figure}
}

Like PBFT and HotStuff, \probft operates in a sequence of views, each having a designated leader responsible for proposing a value.
The protocol consists of two modes of execution \--- normal case and view-change.
The normal case starts when the leader broadcasts its proposal through a \textsc{Propose} message.
Since the leader might be Byzantine and send distinct proposals to different replicas, non-Byzantine replicas need to communicate with each other to check that they received the same proposal.
With this aim, upon receiving a \textsc{Propose} message, a correct replica multicasts the proposal to a sample of $o \times q$ distinct replicas taken uniformly at random from the set of replicas, where $q = O(\sqrt{n})$ and $o > 1$ is a real constant.
Upon receiving \textsc{Prepare} messages from a probabilistic quorum with size~$q$, a correct replica multicasts a \textsc{Commit} message to another random sample composed of $o \times q$ distinct replicas.
Upon receiving \textsc{Commit} messages from a probabilistic quorum with size $q$, a correct replica decides on the proposed value.

\probft's normal case execution relies on probabilistic quorums to solve consensus. 
That is, in contrast to traditional BFT protocols (e.g., PBFT, HotStuff, and randomized protocols like~\cite{cachin2000random,mostefaoui2015signature}), we abandon the requirement of quorums strictly having to intersect and instead only aim at quorums overlapping with high probability. 
As a key benefit, this strategy enables us to keep the number of communication steps at a minimum while significantly reducing quorum sizes, thereby improving resource consumption and scalability. 
Specifically, for a system with $n$~replicas, \probft employs probabilistic quorums of size $q=l\sqrt{n}$, with $l \geq 1$ being a configurable, typically small constant~\cite{malkhi01probabilistic}. 
For example, for $l=2$ and $n=100$, a replica can make progress after receiving~$20$ matching messages from different replicas, which is a significant reduction when compared with the $67$~messages necessary in PBFT.

Since a comparably small number of messages is sufficient to advance a phase in \probft, to offer resilience against Byzantine behavior, it is crucial to prevent faulty replicas from manipulating the decisions in probabilistic quorums~(e.g.,~by flooding the system with their own messages).
In \probft, this is achieved by delegating the selection of message recipients to a globally known \vrf.
That is, in the protocol phases relying on probabilistic quorums, replicas do not freely pick the destinations of their messages but instead are required to send the messages to the specific group of recipients determined by the \vrf. 

Although novel, \probft is heavily based on PBFT, being thus somewhat simple to understand and implement.
Nonetheless, the probabilistic nature of the protocol makes its analysis far from trivial.
More specifically, the main challenges encountered in analyzing \probft are:

\begin{itemize}[leftmargin=1em,label=--]
\item In the analysis of probabilistic algorithms, we often deal with \emph{independent} events, enabling the application of powerful and more straightforward bounds like the Chernoff bounds~\cite{motwani1995randomized}. 
However, in \probft, the probability of forming probabilistic quorums by replicas is \emph{dependent}.
That is, as replicas multicast their \textsc{Prepare} (resp., \textsc{Commit}) messages to random samples, knowing that a replica has received \textsc{Prepare} (resp., \textsc{Commit}) messages from a probabilistic quorum decreases the chance of other replicas to receive \textsc{Prepare} (resp., \textsc{Commit}) messages from a probabilistic quorum. 
This dependency complicates \probft's analysis.

\item The probability of deciding a value by a replica depends on the number of replicas that multicast matching \textsc{Commit} messages.
Additionally, the number of replicas that multicast their \textsc{Commit} messages depends on the number of replicas that multicast their \textsc{Prepare} messages.
This dual dependency layer for computing the probability of deciding a value by a replica adds complexity to the analysis of \probft.
\item A Byzantine leader might send multiple proposals to violate safety.
Rather than examining each possible case individually, we find the optimal behavior for a Byzantine leader, considering that it intends to maximize the probability of safety violation.
\end{itemize}

In summary, besides the design of \probft, the main technical contribution of this paper is the analysis of the protocol, represented by the following theorem:

\begin{theorem}[Informal main result]
\probft guarantees liveness with probability $1$ and safety with a probability of $1 - \mathit{exp}(-\Theta(\sqrt{n}))$.
\end{theorem}

\paragraph{Paper organization}
The remainder of the paper is organized as follows. 
Section~\ref{sec:preliminaries} introduces our system model and describes the key techniques employed in \probft.
Section~\ref{sec:probft} describes the \probft protocol.
Section~\ref{sec:correctness:proofs} presents the correctness proofs of \probft.
Section~\ref{sec:evaluation} presents a numerical analysis of the protocol.
Finally, Sections~\ref{sec:related:work} and~\ref{sec:conclusion} discuss related work and conclude the paper, respectively.

\section{Preliminaries}
\label{sec:preliminaries}

\subsection{System Model}

We consider a distributed system composed of a finite set~$\Pi$ of $n$~processes, which we call replicas, among which up to $f < n/3$ might be subject to Byzantine failures~\cite{lamport_1982} and not behave according to the protocol specification. 
A non-faulty replica is said to be \textit{correct}. 
During execution, we denote by $\Pi_C$ and $\Pi_F$ the sets of correct and faulty replicas, respectively.
The system is partially synchronous~\cite{dwork_1988,makingByzConsLive} in which the network and replicas may operate asynchronously until some \emph{unknown global stabilization time} GST, after which the system becomes synchronous, with \textit{unknown time bounds for communication and computation}.

We assume that each replica has a unique ID, and it is infeasible for a faulty replica to obtain additional IDs to launch a \emph{Sybil attack} \cite{douceur_2002}.
We consider a \emph{static corruption adversary}, i.e., $\Pi_F$ is fixed at the beginning of execution by the adversary.
Byzantine replicas may collude and coordinate their actions.
It is important to note that while Byzantine replicas may be aware of $\Pi_F$, the correct replicas are unaware of $\Pi_F$ and only know the value of $f$. 
Furthermore, we assume an adversarial scheduler that manipulates the delivery time of messages independent of the sender's identifier, its past and current states, and whether it is Byzantine or not.

Each replica signs outgoing messages with its private key and only processes an incoming message if the message's signature can be verified using the sender's public key. 
We denote \msg{T}{m}{i} as a message of type~\textsc{T} with content~$m$ signed by replica~$i$.
We assume that the distribution of keys is performed before the system starts. 
At run-time, the private key of a correct replica never leaves the replica and, therefore, remains unknown to faulty replicas. In contrast, faulty replicas might learn the private keys of other faulty replicas.
In practical settings, it is a standard assumption that the adversary does not have unlimited computational resources; therefore, they cannot break cryptographic primitives. 

\subsection{Consensus}\label{subsec:consensus:properties}
We assume there is an application-specific $\mathtt{valid}$ predicate to indicate whether a value is acceptable~\cite{makingByzConsLive,valid,modsmart}.
Assuming that each correct replica proposes a valid value, any protocol that solves probabilistic consensus satisfies the following properties:

\begin{itemize}[leftmargin=1em,label=--]
\item \textbf{Validity.} The value decided by a correct replica satisfies the application-defined \texttt{valid} predicate.
\item \textbf{Probabilistic Agreement.} Any two correct replicas decide on different values with probability~$\rho$ depending on the number of existing Byzantine replicas and quorum/sample sizes.
\item \textbf{Probabilistic Termination.} Every correct replica decides with probability $1$.
\end{itemize}

The first two properties are safety properties, while the last one is a liveness property.

\subsection{Single-shot PBFT}
Since \probft follows a structure very similar to PBFT~\cite{pbft}, we briefly review the single-shot version of the latter considering the use of a \emph{synchronizer}~\cite{makingByzConsLive}.
This version is a leader-based consensus protocol that operates in a succession of views produced by the synchronizer, each having a designated leader defined in a round-robin way.
Each protocol view consists of three steps and works as follows (see Figure~\ref{fig:PBFT:overview}):

\begin{itemize}[leftmargin=1em,label=--]
    \item \textbf{\textsc{Propose} (or pre-prepare) phase.}
    The leader of view~$v$ is responsible for proposing a value to the other replicas.
    With this aim, the leader broadcasts a value through a \textsc{Propose} message.
    A correct leader must carefully choose the value to ensure that if a correct replica has decided on a value in a previous view, it will propose the same value.
    \item \textbf{\textsc{Prepare} phase.} 
    Upon receiving a \textsc{Propose} message from a replica $i$ in view $v$, a correct replica broadcasts a \textsc{Prepare} message if $i$ is the leader of $v$, and the proposed value is a valid proposal.
    \item \textbf{\textsc{Commit} phase.}
    Upon receiving \textsc{Prepare} messages from a quorum of replicas, a correct replica broadcasts a \textsc{Commit} message.
\end{itemize}
A replica decides the value proposed by the leader upon receiving \textsc{Commit} messages from a quorum of replicas.
There are multiple scenarios where replicas cannot decide a value in a view~$v$, like when the leader is Byzantine and remains silent.
In those scenarios, the synchronizer transitions the view from $v$ to $v+1$, changing the designated leader.

\subsection{Verifiable Random Function}\label{sec:vrf}
A \emph{verifiable random function} (\vrf)~\cite{algorand,goldberg22verifiable} enables the random selection of a subset from a given set, ensuring that the selection process is verifiable and secure.
We assume a globally known \vrf that provides the following two operations:

\begin{figure}[!t]
    \centering
    \includegraphics[scale=1.4]{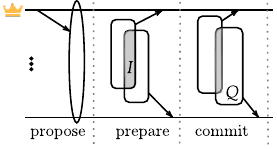}
    \caption{Overview of PBFT.  
    Each correct replica broadcasts its \textsc{Prepare} and \textsc{Commit} messages.
    The size of any quorum is $|Q| = \lceil{(n+f+1)/2}\rceil$.
    The set of replicas $I$ in the intersection of two quorums contains at least one correct replica.}
    \label{fig:PBFT:overview}
    \Description[]{Overview of PBFT.  
    Each correct replica broadcasts its \textsc{Prepare} and \textsc{Commit} messages.
    The size of any quorum is $|Q| = \lceil{(n+f+1)/2}\rceil$.
    The set of replicas $I$ in the intersection of two quorums contains at least one correct replica.}
\end{figure}

\begin{itemize}[leftmargin=1em,label=--]
	\item \textbf{\texttt{VRF\_prove($K_{p,i}$,\,$z$,\,$s$)} $\Rightarrow$ $S_i$,\,$P_i$.}
    Given the private key~$K_{p,i}$ for a replica~$i$, a seed~$z$, and a positive integer $s$, \texttt{VRF\_prove} selects a sample~$S_i$ containing the IDs of $s$ distinct replicas uniformly at random.  
    Along with $S_i$, this operation returns a proof~$P_i$, enabling other replicas to verify whether the sample was obtained using this operation.
    \item \textbf{\texttt{VRF\_verify($K_{u,i}$,\,$z$,\,$s$,\,$S_i$,\,$P_i$)} $\Rightarrow$ \texttt{bool}.} 
    Given the public key $K_{u,i}$ of replica~$i$, a seed~$z$, a positive integer $s$, a sample~$S_i$, and its associated proof $P_i$, \texttt{VRF\_verify} determines whether $S_i$ is a valid sample generated using \texttt{VRF\_prove} with the given parameters.
    It returns \texttt{true} if the sample and proof are valid and \texttt{false} otherwise.
\end{itemize}

The \vrf should provide the following guarantees~\cite{goldberg22verifiable}: 

\begin{itemize}[leftmargin=1em,label=--]
    \item \textbf{Uniqueness.} 
    A computationally limited adversary must not be able to produce two different proofs~$P_{i}$ and $P_{i}'$ for the same input parameters $K_{u,i}$, $z$, and $s$. 
    \item \textbf{Collision resistance.} 
    Even when a private key is compromised, it should be infeasible for an adversary to find two distinct seeds~$z$ and $z'$ for which the \texttt{VRF\_prove} returns the same sample. 
    \item \textbf{Pseudorandomness.} 
    For an adversarial verifier without knowledge of the proof, the corresponding sample should be indistinguishable from a randomly selected set of replica IDs.
\end{itemize}

\section{\probft}\label{sec:probft}

\begin{figure}[!t]
    \centering
    \includegraphics[scale=1.36]{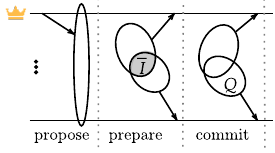}
    \caption{Overview of \probft. 
    The size of any probabilistic quorum is $q = |Q| = O(\sqrt{n})$. 
    Each correct replica multicasts its \textsc{Prepare} and \textsc{Commit} messages to random samples of sizes $o\times q$, where $o$ is a constant.
    The set of replicas $\bar{I}$ in the intersection of two probabilistic quorums contains at least one correct replica with high probability.}
    \label{fig:proBft:overview}
    \Description[]{Overview of \probft. 
    The size of any probabilistic quorum is $q = |Q| = O(\sqrt{n})$. 
    Each correct replica multicasts its \textsc{Prepare} and \textsc{Commit} messages to random samples of sizes $o\times q$, where $o$ is a constant.
    The set of replicas $\bar{I}$ in the intersection of two probabilistic quorums contains at least one correct replica with high probability.}
\end{figure}

\subsection{Overview}\label{subsec:probft_overview}

\probft is a leader-based probabilistic consensus protocol that operates in a succession of views produced by the synchronizer.
In \probft, illustrated in Figure~\ref{fig:proBft:overview}, one of the replicas is assigned the role of \emph{leader} in each view, meaning that this replica is in charge of proposing a value. 
\probft's consensus process comprises three phases of message exchange between replicas -- \emph{propose}, \emph{prepare}, and \emph{commit}, just like in PBFT.
In the \emph{propose} phase, the leader proposes a value by broadcasting it to other replicas, which is then agreed on in the subsequent \emph{prepare} and \emph{commit} phases. 

Given the linear message complexity of the \emph{propose} phase, we restrict the utilization of probabilistic quorums to the two remaining phases.
That is, in each of the \emph{prepare} and \emph{commit} phases, any replica relies on the \vrf to determine a subset of replicas with size $o\times q$ to whom its messages should be sent, where $o > 1$ is a constant. 
Any replica progresses by receiving messages from a probabilistic quorum with size $q = l\sqrt{n}$, being $l$ also a small constant (e.g., $2$).
The constant~$o$ defines how large the random subset of replicas contacted on each phase by each replica is when compared with the probabilistic quorum size.
Bigger values of $o$ increase the probability of forming a probabilistic quorum (with $q$ replicas), increasing the chance of the protocol to terminate (see Section~\ref{sec:correctness:proofs}), albeit generating more messages (see Figure~\ref{fig:pbft:probft:hotstuff}).
As a result, in contrast to the $O(n^2)$ communication complexity associated with traditional protocols such as PBFT, \probft's message complexity for the \emph{prepare} and \emph{commit} phases is $O(n\sqrt{n})$.

As seed for computing its recipient sample with the \vrf, a replica is required to use $v$\,$||$\,\textsc{T}, which is a concatenation of the current view~$v$, and an identifier \textsc{T} representing the phase/message type~(``prepare''~for \textsc{Prepare} and ``commit''~for \textsc{Commit}). 
Each replica starts a \emph{prepare} phase upon receiving the leader's proposal and a \emph{commit} phase upon forming a probabilistic quorum, i.e., receiving $q$ matching messages.
Forcing replicas to involve the \vrf this way and thereby apply deterministic seeds has the following key benefits: 
\begin{enumerate}[label=(\arabic*), leftmargin=1.5em]
    \item With the inputs of \texttt{VRF\_prove} being dictated by the protocol, faulty replicas cannot control their recipient sample for a particular view and phase. Consequently, a faulty replica cannot deliberately favor a certain subset of replicas, for example, in support of a faulty leader trying to trick these replicas into forming a probabilistic quorum for a specific value.  
    \item Since a replica's recipient samples are computed from its private key~(see Section~\ref{sec:vrf}), faulty replicas cannot predict the individual samples of correct replicas in advance. 
    That is, at the start of a view, faulty replicas do not know the upcoming \emph{prepare} and \emph{commit} phase message-exchange patterns between correct replicas, thereby making it inherently difficult for them to identify promising attack targets~(e.g.,~correct replicas that are included in none or only in a few recipient samples of other correct replicas). 
    \item With the \emph{prepare} and \emph{commit} phase recipient samples usually differing~(due to the use of the phase parameter in the seed), correct replicas are more likely to observe the misbehavior of a faulty leader. For example, if a faulty leader performs equivocation by proposing different values to different subsets of replicas, then the phase-specific recipient samples increase the probability of a correct replica learning about the existence of contradictory proposals. 
\end{enumerate}

In summary, the use of the \vrf for selecting recipient samples in the \emph{prepare} and \emph{commit} phases significantly strengthens \probft's resilience against malicious behaviors.

\subsection{Protocol Specification}\label{sec:probft_protocol_description}

Algorithm \ref{alg:probft} presents \probft specification. 
This description assumes a synchronizer exactly like the one presented in~\cite{makingByzConsLive}.
\begin{algorithm}[t!]
\caption{\probft \---  replica~$i$.}
\label{alg:probft}
\begin{algorithmic}[1]
\NoThen
\NoDo

\STATEx{\hspace{-1.67em}\textbf{upon} \texttt{newView($v$)}}

\ifbool{extendedVersion}{
    \STATE{$\mathit{curView}, \mathit{curVal}, \mathit{voted}, \mathit{blockView}, \mathit{proposal}\leftarrow v, \bot, \mathtt{false}, \mathtt{false}, \perp$}
}{
    \STATE{$\mathit{curView} \leftarrow v$}
    \STATE{$\mathit{curVal} \leftarrow \bot$}
    \STATE{$\mathit{voted} \leftarrow \mathtt{false}$}
    \STATE{$\mathit{blockView} \leftarrow \mathtt{false}$}
    \STATE{$\mathit{proposal} \leftarrow \perp$}
}

\IF{$\mathit{curView}=1 \land i=\mathtt{leader}(\mathit{curView})$}
    \STATE{\textbf{broadcast} $\langle \textsc{Propose}, \langle \mathit{curView}, \mathtt{myValue}()\rangle_i,\perp\rangle_i$}\label{line:cast:v:1}
\ELSIF{$\mathit{curView} > 1$}
    \ifthenelse{\boolean{extendedVersion}}{
        \STATE{\textbf{send} $\langle \textsc{NewLeader}, \mathit{curView}, \mathit{preparedView},\mathit{preparedVal},\mathit{cert}\rangle_i$ \textbf{to} $\mathtt{leader}(\mathit{curView})$}\label{line:send_new_leader}
    }{
    \STATE{\textbf{send} $\langle \textsc{NewLeader}, \mathit{curView}, \mathit{preparedView},\mathit{preparedVal},\mathit{cert}\rangle_i$ \strut\hfill \textbf{to} $\mathtt{leader}(\mathit{curView})$}\label{line:send_new_leader}
    }
\ENDIF

\vspace{0.2em}
\STATEx{\hspace{-1.67em} \textbf{upon receiving} $\{ \langle \textsc{NewLeader},v,\mathit{view}_j,\mathit{val}_j,\mathit{cert}_j\rangle_j : j \in Q \} = M$ \textbf{from} a deterministic quorum~$Q$}\label{line:new_leader_quorum}
\ifthenelse{\boolean{extendedVersion}}{
    \STATE{\textbf{pre:} $\mathit{curView} = v \land i =  \mathtt{leader}(v) \land (\forall m \in M: \mathtt{validNewLeader}(m))$}
}{
    \STATE{\textbf{pre:} $\mathit{curView} = v \land i =  \mathtt{leader}(v) \ \land $}
    \STATEx{$(\forall m \in M: \mathtt{validNewLeader}(m))$}
}
\STATE{$v_\mathit{max} \leftarrow \mathtt{max}\{\mathit{view}_j : \langle \textsc{NewLeader},v,\mathit{view}_j, \_, \_ \rangle_j \in M$\}}\label{line:v:max}
\STATE{$\mathit{val}_\mathit{max} = \mathtt{mode} \{ \mathit{val}_j :\langle \textsc{NewLeader},v,v_\mathit{max}, \mathit{val}_j, \_ \rangle_j \in M\}$}\label{line:val:max}
\IF{$\mathit{val}_\mathit{max} \neq \bot$}\label{line:most_prepared}
    \STATE{\textbf{broadcast} $\langle \textsc{Propose}, \langle v, \mathit{val}_\mathit{max}\rangle_i,M\rangle_i$}
\ELSE
    \STATE{\textbf{broadcast} $\langle \textsc{Propose}, \langle v, \mathtt{myValue}()\rangle_i,M\rangle_i$}\label{line:propose:myVal}
\ENDIF

\vspace{0.2em}
\STATEx{\hspace{-1.67em} \textbf{upon receiving} $\langle \textsc{Propose}, \langle v,x \rangle_j ,\_\rangle_j = m$}
\STATE{\textbf{pre:} $\mathit{blockView} = \mathtt{false} \land \mathit{curView} = v \land \mathit{voted} = \mathtt{false} \land \mathtt{safeProposal}(m)$}\label{line:safe:proposal}
\STATE{$\mathit{curVal},\mathit{voted}, \mathit{proposal} \leftarrow x,\mathtt{true}, m$}\label{line:set:currVal}
\STATE{$S_p,P_p \leftarrow \mathtt{VRF\_prove}(K_{p,i},v \;||\;``\text{prepare}", o\times q)$}\label{line:sample_prepare}
\STATE{\textbf{send} $\langle \textsc{Prepare}, \langle v,x \rangle_j, S_p, P_p\rangle_i$ \textbf{to} $S_p$}\label{line:send_prepare}

\vspace{0.2em}
\STATEx{\hspace{-1.67em} \textbf{upon receiving} $\{ \langle \textsc{Prepare},\langle v,x \rangle_*,S,P\rangle_j : j \in Q\} = C$ \textbf{from} a probabilistic quorum $Q$}\label{line:observe_prepare_quorum}
\ifbool{extendedVersion}{
    \STATE{\textbf{pre:} $\mathit{blockView} = \mathtt{false} \land \mathit{curView} = v \land \mathit{curVal} = x \land \mathit{voted} = \mathtt{true} \ \land$} 
    \STATEx{$(\forall \langle \_,\_,S,P \rangle_j \in C: i \in S \land \mathtt{VRF\_verify}(K_{u,j},v \ || \ ``\text{prepare}",o\times q,S,P))$}
}{
    \STATE{\textbf{pre:} $\mathit{blockView} = \mathtt{false} \land \mathit{curView} = v \land \mathit{curVal} = x \ \land$} 
    \STATEx{$\mathit{voted} = \mathtt{true} \ \land (\forall \langle \_,\_,S,P \rangle_j \in C: i \in S \ \land$}
    \STATEx{$\mathtt{VRF\_verify}(K_{u,j},v \ || \ ``\text{prepare}",o\times q,S,P))$}
}
\STATE{$\mathit{preparedVal},\mathit{preparedView},\mathit{cert} \leftarrow\mathit{curVal},\mathit{curView},C$}\label{line:prepared_val}
\STATE{$S_c,P_c \leftarrow  \mathtt{VRF\_prove}(K_{p,i},v \ || \  ``\text{commit}", o\times q)$}
\STATE{\textbf{send} $\langle \textsc{Commit},\langle v,x \rangle_*,S_c,P_c\rangle_i$ \textbf{to} $S_c$}\label{line_send_commit}

\vspace{0.2em}
\STATEx{\hspace{-1.67em} \textbf{upon receiving} $\{ \langle \textsc{Commit},\langle v,x \rangle_*,S,P\rangle_j : j \in Q\} = M$ \textbf{from} a probabilistic quorum $Q$}\label{line:observe_commit_quorum}
\ifbool{extendedVersion}{
    \STATE{\textbf{pre:} $\mathit{blockView} = \mathtt{false} \land \mathit{preparedVal} = x \land \mathit{curView} = \mathit{preparedView} = v \ \land$}\label{line:process:commit}
    \STATEx{$(\forall \langle \_,\_,S,P \rangle_j \in M: i \in S \land \mathtt{VRF\_verify}(K_{u,j},v \ || \ ``\text{commit}",o\times q,S,P))$}
}{
    \STATE{\textbf{pre:} $\mathit{blockView} = \mathtt{false} \land \mathit{preparedVal} = x \ \land$}\label{line:process:commit}
    \STATEx{$\mathit{curView} = \mathit{preparedView} = v \land (\forall \langle \_,\_,S,P \rangle_j \in M: i \in S \ \land$}
    \STATEx{$\mathtt{VRF\_verify}(K_{u,j},v \ || \ ``\text{commit}",o\times q,S,P))$}
}
\STATE{$\mathtt{decide}(\mathit{\mathit{curVal}})$}\label{line:decide}

\vspace{0.2em}
\STATEx{\hspace{-1.67em} \textbf{upon receiving} $\langle \_,\langle v, x \rangle_j,\dots\rangle_* = m$ }
\STATE{\textbf{pre:} $\mathit{blockView} = \mathtt{false} \land \mathit{curView} = v   \land j = \mathtt{leader}(v) \land \mathit{voted} = \mathtt{true} \land \mathit{curVal} \neq x $} \label{line:block_view_start}
\STATE{$\mathit{blockView} \leftarrow \mathtt{true}$}\label{line:block_view_end}
\STATE{\textbf{broadcast} $m$, \textbf{broadcast} $\mathit{proposal}$}\label{line:notify:faultiness}

\end{algorithmic}
\end{algorithm}
The protocol proceeds in a series of views, with each new view having a fixed leader responsible for proposing a value to be decided. 
Every replica in the system can determine the leader for a view~$v$ with the $\mathtt{leader}$ predicate. 
\begin{align*}
    \begin{split}
    \mathtt{leader}(v) = (v-1 \ \mathit{mod}\ n) +1	
    \end{split}
\end{align*}
Upon receiving a notification from the synchronizer to transition to view~$v$, a replica stores~$v$ in a variable $\mathit{curView}$ and sets a flag $\mathit{voted}$ to \texttt{false} to record that it has not yet received any proposal from the leader in the current view.
If $v=1$, the leader is free to broadcast its proposal (line~\ref{line:cast:v:1}).
However, for other views, a correct leader must be careful in choosing its proposal because if a correct replica has already decided on a value in a prior view, the leader is obligated to propose the same value.
To facilitate this, upon entering a view~$v>1$, a correct replica sends a \textsc{NewLeader} message to the leader of~$v$, providing information about the latest value it accepted in a prior view (line~\ref{line:send_new_leader}).
Any message exchanged in the protocol is tagged with the sender's view. 
A receiver will only accept a message if its own view stored in the $\mathit{curView}$ variable matches the view of the sender.

For any view $v>1$, the leader of~$v$ waits until it receives \textsc{New\-Leader} messages from a \emph{deterministic quorum} of replicas.
After computing its proposal, the leader broadcasts the proposal, along with some supporting information, in a \textsc{Propose} message (lines~\ref{line:v:max}-\ref{line:propose:myVal}). 
After presenting the rest of the protocol, we will describe the process for computing the proposal.
Since a Byzantine leader may send different proposals to different replicas, correct replicas need to communicate with others to ensure they have received the same proposal.
With this aim, correct replicas process the leader's proposal in two phases~--~\emph{prepare} and \emph{commit}.

Upon receiving a \textsc{Propose} message~$m$ containing a proposal $x$ from a replica~$j$ in view $v$, a correct replica~$i$ starts the \emph{prepare} phase if it is currently in view~$v$, it has not processed a \textsc{Propose} message in this view, and message~$m$ satisfies the $\mathtt{safeProposal}$ predicate (also explained later), which ensures that a Byzantine leader cannot reverse decisions reached in a prior view (with high probability). 
The replica then stores $x$ in $\mathit{curVal}$ and sets voted to \texttt{true} (line~\ref{line:set:currVal}).
Afterward, replica~$i$ uses the \vrf to select a random sample $S_P$ to which it sends a \textsc{Prepare} message.
 
A correct replica waits until receiving a set $C$ of \textsc{Prepare} messages from a probabilistic quorum.
We call this set of messages a \textit{prepared certificate} for a proposed value~$x$ in a view~$v$ if it satisfies the following predicate:
\begin{align*}
& \mathtt{prepared}\left(C, v, x, j \right) \iff \nonumber
\\&\quad \exists Q : |Q|=q \land C = \{\langle \textsc{Prepare}, \langle v, x \rangle_i, S_k, P_k \rangle_k : k \in Q\} \ \land \nonumber
\\&\quad i = \mathtt{leader}(v) \land (\forall \langle \_, \_, S_k, P_k \rangle_k \in C : j \in S_k \ \land
\\&\quad \quad  \mathtt{VRF\_verify}(K_{u,k}, v \ || \ ``\text{prepare}", o\times q, S_k, P_k))\nonumber
\end{align*}
Once a replica $j$ creates a prepared certificate for a value $x$ in a view $v$ (i.e., $j$ prepares~$x$), it stores $x$, $v$, and this certificate in $\mathit{preparedVal}$, $\mathit{preparedView}$, and $\mathit{cert}$, respectively. 
Afterward, the replica generates a new random sample of replicas $S_c$ to whom it will multicast a $\textsc{Commit}$ message (lines~\ref{line:prepared_val}-\ref{line_send_commit}). 
Every correct replica that multicasts a $\textsc{Commit}$ message enters the \emph{commit} phase.
It then waits until receiving $\textsc{Commit}$ messages from a probabilistic quorum. 
It is worth noting that a correct replica neither sends a $\textsc{Commit}$ message nor processes a received $\textsc{Commit}$ message (line~\ref{line:process:commit}) if it has not yet prepared a value.
After observing a quorum of $\textsc{Commit}$ messages, a correct replica with a prepared certificate decides on the proposed value (line~\ref{line:decide}).


Recall that when the synchronizer triggers a $\mathtt{newView}$ notification in a replica for a view greater than one, the replica sends a \textsc{NewLeader} message to the new leader.
If a replica has created a prepared certificate in a prior view, it sends that certificate in the \textsc{NewLeader} message.
This allows the leader to generate its proposal based on a quorum of well-formed \textsc{NewLeader} messages that can be checked using the following predicate:
\begin{align*}
\begin{split}
& \mathtt{validNewLeader}(\langle \textsc{NewLeader}, v, \mathit{view}, \mathit{val}, \mathit{cert} \rangle_j) \iff 
\\ & \qquad \mathit{view} < v \land \mathit{view} \neq 0 \Rightarrow \mathtt{prepared}(\mathit{cert}, \mathit{view}, \mathit{val}, j) 
\end{split}
\end{align*}

The leader chooses its proposal by selecting the value prepared in the most recent view by more replicas (lines~\ref{line:v:max}-\ref{line:val:max}).
If no such prepared values exist, it uses its own proposal provided by the function \texttt{myValue()}.
Since a faulty leader may not follow this rule, it is essential for the correct replicas to validate that the leader adheres to this selection rule for proposals.
With this aim, the leader's \textsc{Propose} message contains the \textsc{NewLeader} messages received by the leader, in addition
to the proposal. 
A correct replica checks the validity of the proposed value by redoing the leader's computation using the following predicate:
\begin{align*}
\begin{split}
    & \mathtt{safeProposal}(\langle \textsc{Propose},\langle v,x \rangle_j ,M\rangle_j) \iff  
    \\&\quad v \geq 1 \land j = \mathtt{leader}(v) \land \mathtt{valid}(x) \ \land  ( v = 1 \ \lor
    \\&\quad \quad (|M| \geq \lceil{(n+f+1)/2}\rceil \land  (\forall m \in M:\; \mathtt{validNewLeader}(m) ) \land 
    \\&\quad \quad \quad (\exists v_\mathit{max} = \mathtt{max}\{\mathit{view}_k : \langle \textsc{NewLeader},v,\mathit{view}_k, \_, \_ \rangle_k \in M\} \land
    \\&\quad \quad \quad \quad x = \mathtt{mode} \{ \mathit{val}_k :\langle \textsc{NewLeader},v,v_\mathit{max}, \mathit{val}_k, \_ \rangle_k \in M\}))) 
\end{split}
\end{align*}

When a correct replica detects that the leader is faulty, i.e., receiving messages from any replica with different proposals signed by the leader, it instantly blocks the current view and waits for the synchronizer to trigger a new view (lines~\ref{line:block_view_start}-\ref{line:notify:faultiness}). 
Besides, it informs other replicas about this misbehavior.
It is important to emphasize that informing other replicas is necessary only when the leader is Byzantine and sends distinct proposals.
Therefore, it does not impact the message complexity of the protocol when the leader is correct.

\subsection{Message and Communication Complexities}
\probft's message complexity is $O(n\sqrt{n})$, as computed based on four terms:~$O(n)$ for \textsc{NewLeader} messages, $O(n)$ for \textsc{Propose} messages, $O(n\sqrt{n})$ for \textsc{Prepare} messages, and $O(n\sqrt{n})$ for \textsc{Commit} messages.
Further, in \probft, for any view greater than one, a new leader sends a \textsc{Propose} message with a certificate containing a full (not probabilistic) quorum of \textsc{NewLeader} messages to all replicas.
Each \textsc{NewLeader} message might contain a prepared certificate with a probabilistic quorum of \textsc{Prepare} messages. 
Hence, \probft's communication complexity is $O(n^2\sqrt{n})$.
Note that \probft has this communication complexity only when a view-change occurs. 
In the first view, there is no need to send \textsc{NewLeader} messages to the leader, as avoided in practical instantiations of PBFT (e.g.,~\cite{pbft,bessani2014state}). 
Therefore, \probft's best-case communication complexity is $\Omega(n\sqrt{n})$, contrary to PBFT, which still has $\Omega(n^2)$. 

\section{\probft Proof Outline}\label{sec:correctness:proofs}
This section outlines the correctness proofs of \probft.
More specifically, we show that \probft satisfies three properties of probabilistic consensus, i.e., Validity, Probabilistic Termination, and Probabilistic Agreement. 
\ifthenelse{\boolean{extendedVersion}}{
    Here we discuss the main arguments of the proofs and refer to the appendix for the theorems whose proofs are not provided in this section.
}{
    Here we discuss the main arguments of the proofs and refer to the extended version~\cite{extended} for the full proofs.
}

\subsection{Validity}
Recall that the Validity property states that the value decided by a correct replica satisfies the application-defined $\mathtt{valid}$ predicate.
In \probft, before deciding a value $x$, a correct replica $i$ must receive~$x$ as a proposal from the leader and verify its validity using the $\mathtt{safeProposal}$ predicate (line~\ref{line:safe:proposal}).
In this predicate, several conditions must be satisfied, one of which is that $\mathtt{valid}(x)$ must be $\mathit{true}$.
Accordingly, a value decided by a correct replica is valid, satisfying the Validity property.

\subsection{Probabilistic Termination} 
To demonstrate that \probft satisfies Probabilistic Termination, suppose there is a non-empty subset of correct replicas that have not decided a value by GST.
Recall that \probft employs a round-robin mechanism for changing the leaders.
Accordingly, among the replicas that have not made a decision by GST, there will be some correct replica $i$ that performs the leader's role by proposing a value~$x$ in a view~$v$.
Note that all correct replicas receive such a proposal during~$v$ as the system is synchronous after GST and multicast their \textsc{Prepare} messages.
We show that any correct replica decides~$x$ with a high probability during view $v$ if it has not already decided a value.
To do so, we first demonstrate that any correct replica receives \textsc{Prepare} messages from a probabilistic quorum of replicas with a high probability.
The following theorem demonstrates that with a proper value of $o$, such an event occurs for a correct replica with a probability of at least $1 - \mathit{exp}(-\sqrt{n})$, i.e., forming such a quorum with high probability, even if all Byzantine replicas remain silent. 
\begin{theorem}\label{cor:epsilon} 
Suppose each correct replica takes a sample composed of $o \times l\sqrt{n}$ distinct replicas, uniformly at random from $\Pi$, and multicasts a message to them, where $ o \in [1, 3.732(n/(n-f))] $, and $l\geq 1$.
Provided that a replica forms a probabilistic quorum upon receiving $l\sqrt{n}$ messages, the probability of forming such a quorum is at least $1 - \mathit{exp}(- \sqrt{n} )$.  
\end{theorem}

Recall that a correct replica prepares the value proposed by the leader upon receiving \textsc{Prepapre} messages from a probabilistic quorum.
According to Theorem~\ref{cor:epsilon}, a correct replica prepares the value proposed by a correct leader with high probability, so almost all correct replicas send \textsc{Commit} messages to their randomly selected samples.
It is clear that the probability of receiving \textsc{Commit} messages from a probabilistic quorum is less than the probability of receiving \textsc{Prepare} messages from a probabilistic quorum.
However, the following theorem demonstrates that every correct replica receives \textsc{Commit} messages from a probabilistic quorum with a high probability, resulting in deciding a value with a high probability.
\begin{theorem}\label{cor:termination}
After GST, if the leader of view $v$ is correct, then each correct replica decides a value in $v$ with a probability of at least $1 - 2(n-f)\times\mathit{exp}\boldsymbol{(}-\Theta(\sqrt{n})\boldsymbol{)}$.
\end{theorem}

Proving the above theorem constitutes the most complex part of demonstrating that \probft satisfies the Probabilistic Termination property.
This complexity arises from the following sources:
\begin{itemize}[leftmargin=1em,label=--]
    \item  The probability of forming probabilistic quorums by replicas is \textit{dependent}.
    That is, as replicas multicast their \textsc{Prepare} (resp., \textsc{Commit})  messages to random samples, knowing that a replica has received \textsc{Prepare} (resp., \textsc{Commit}) messages from a probabilistic quorum decreases the chance of other replicas to receive \textsc{Prepare} (resp., \textsc{Commit}) messages from a probabilistic quorum. 
    This dependency prevents us from directly using well-known and sharp bounds like the Chernoff bounds~\cite{motwani1995randomized}. 
    To circumvent this dependency, we use the notion of negative association~\cite{concentration}, enabling us to leverage the Chernoff bounds.
    \item The number of replicas that receive \textsc{Commit} messages from probabilistic quorums \emph{depends on} the number of replicas that receive \textsc{Prepare} messages from probabilistic quorums.
    We address this dual dependency inherent in computing the probability of deciding a value by a replica by conditioning the probability of forming quorums by \textsc{Commit} messages on the probability of forming quorums by \textsc{Prepare} messages.
\end{itemize}

In \probft, when the leader of view $v$ is correct, a correct replica might not receive enough messages to form quorums, leading to not deciding a value in view $v$.
According to Theorem~\ref{cor:termination}, such an event happens with a low probability.
However, as there are infinite views whose leaders are correct, each correct replica decides with probability $1$.
\begin{theorem}[Main liveness result]\label{thm:termination:infty}
In \probft, every correct replica eventually decides a value with probability $1$.
\end{theorem}

\subsection{Probabilistic Agreement}
In \probft, different replicas may decide different values since quorum intersections are not guaranteed, but the protocol has to ensure that the probability of agreement violation is low.
We begin by computing the probability of ensuring agreement within a view.

\paragraph{Probabilistic Agreement in a view}
In \probft, to cause disagreement in a view, it is required that multiple values are decided in the same view by multiple correct replicas.
Such a situation only happens when the leader is Byzantine since correct leaders send a single proposal in their views.

There are many cases in which a Byzantine leader can compromise the agreement in a view.
Rather than examining every possible case individually, we find the optimal behavior for a Byzantine leader, considering that it intends to maximize the probability of agreement violation.
For this purpose, we consider the three cases illustrated in Figure~\ref{fig:casesview}.
The first case is the most general one, which can be used to derive any possible situation.
The second and third cases demonstrate specific situations.
In the third case, the probability of agreement violation is greater than or equal to the probability of agreement violation in the second case and any other situation obtained from the first case.

\begin{figure}[!t]
    \centering
    \begin{subfigure}[t]{0.48\textwidth}
        \centering
        \includegraphics[scale=0.87]{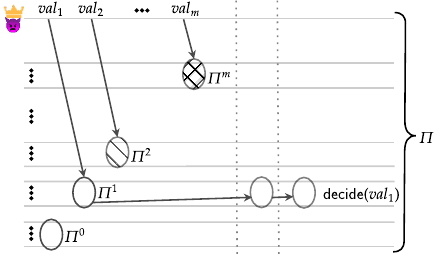}
        \caption{The general case.
        The Byzantine leader sends $m$ different proposals to $m$ non-empty subset of replicas, which might overlap.
        It also does not send any proposal to subset $\Pi^0$.}
        \label{fig:1}
    \end{subfigure}
    \hfill
    \begin{subfigure}[t]{0.48\textwidth}
        \centering
        \includegraphics[scale=0.87]{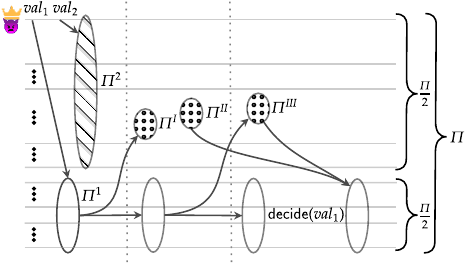}
        \caption{A sub-optimal case.
        The Byzantine leader sends two proposals $\mathit{val}_1$ and $\mathit{val}_2$ to $\Pi^1$ and $\Pi^2$, respectively.}
        \label{fig:2}
    \end{subfigure}
    \\
    \begin{subfigure}[t]{0.48\textwidth}
        \centering
        \includegraphics[scale=0.87]{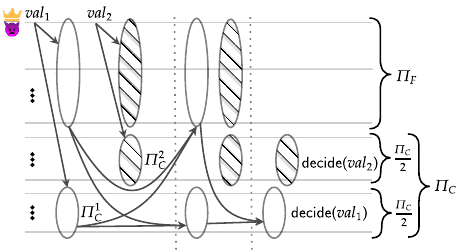}
        \caption{The optimal case.
        Given two sets $\Pi^1_C , \Pi^2_C \subseteq \Pi_C$ with equal sizes, the Byzantine leader sends only two proposals $\mathit{val}_1$ and $\mathit{val}_2$ to $\Pi^1_C \cup \Pi_F$ and $\Pi^2_C \cup \Pi_F$, respectively.}
        \label{fig:3}
    \end{subfigure}
    \caption{Different scenarios in which a Byzantine leader can cause disagreements in a view.}
    \Description[]{}
    \label{fig:casesview}
\end{figure}

\begin{itemize}[label=--,leftmargin=1em]
\item \emph{The general case.}    
    The Byzantine leader sends multiple proposals $\mathit{val}_1, \dots, \mathit{val}_m$, where $m\geq 2$.
    As a result, some replica(s) might receive one or even multiple proposals, and some others might not receive any proposals.
    This case is depicted in Figure~\ref{fig:1}.
\item \emph{A sub-optimal case.} 
    The Byzantine leader divides the set of replicas into two equally sized groups $\Pi^1$ and $\Pi^2$.
    It sends a proposal~$\mathit{val}_1$ to $\Pi^1$ and another proposal $\mathit{val}_2$ to $\Pi^2$.
    This case is depicted in Figure~\ref{fig:2}.
\item \emph{The optimal case.}
    The Byzantine leader makes a distinction between correct and Byzantine replicas.
    It divides the correct replicas into two equally sized groups -- $\Pi_C^1$ and $\Pi_C^2$.
    It sends a proposal $\mathit{val}_1$ to $\Pi_C^1 \cup \Pi_F$ and another proposal $\mathit{val}_2$ to $\Pi_C^2 \cup \Pi_F$.
    This case is depicted in Figure~\ref{fig:3}.
\end{itemize}

In order to describe why the first case represents the most general situation, we need to present some notations. 
For each proposal~$\mathit{val}_i$, where $1 \leq i \leq m$, we associate a set $\Pi^{i} \subset \Pi$ containing each replica~$p$ that receives $\mathit{val}_i$ as a proposal from the leader.
As the leader might send multiple proposals to a replica, any set $\Pi^i$ might intersect with another set $\Pi^j$, where $1\leq i<j\leq m$.   
We denote by $\Pi^0$ the set of replicas to which the leader does not send any proposal, but these replicas can receive messages from other replicas.
Indeed, $\Pi \setminus \Pi^0$ contains every replica that can multicast \textsc{Prepare} and \textsc{Commit} messages. 
Furthermore, we denote the correct (resp. Byzantine) replicas within a set~$\Pi^i$ by~$\Pi^i_C$ (resp. $\Pi^i_F$).

It is essential to remark that a correct replica $p$ to decide a value $\mathit{val}_i$ requires to form a probabilistic prepare quorum $P$ and a probabilistic commit quorum $Q$ such that $P, Q \subseteq \Pi^i$; otherwise, $p$ does not decide a value due to receiving two distinct values.
Consequently, if replica $p \in \Pi^i$ decides a value, the value is $\mathit{val}_i$.

Using the first case, we can model any situation in the system when the Byzantine leader sends multiple proposals, as we do not impose any restrictions on the leader's behavior.
For instance, replica $p$ decides $\mathit{val}_i$ when either 
\begin{enumerate*}[label=(\arabic*)]
\item it forms a probabilistic prepare quorum and a probabilistic commit quorum, both composed of Byzantine replicas within $\Pi^i$, i.e., $P,Q\subseteq \Pi^i_F$,
\item it forms a probabilistic prepare quorum composed of Byzantine replicas and a probabilistic commit quorum composed of correct replicas, i.e., $P \subseteq \Pi^i_F$ and $Q\subseteq \Pi^i_C$, or
\item it forms a probabilistic prepare quorum composed of correct and Byzantine replicas and a probabilistic commit quorum composed of correct replicas, i.e.,
$P \subseteq \Pi^i$, $P \cap \Pi_C^i \cap \Pi_F^i \neq \emptyset$, and $Q\subseteq \Pi^i_C$.
\end{enumerate*}

Beyond the situations where a replica decides a value, the first case demonstrates situations where a replica does not decide a value because of receiving multiple proposals.
For example, replica $p$ does not decide a value when either
\begin{enumerate*}[label=(\arabic*)]
\item it receives multiple proposals from the leader, i.e., $p \in \Pi^i$ and $p \in \Pi^j$, where $1\leq i<j\leq m$, 
\item it receives at least two different proposals $\mathit{val}_i$ and $\mathit{val}_j$ such that $\mathit{val}_i$ is received from the leader and $\mathit{val}_j$ is received from one of the replicas sending its \textsc{Prepare} message, or 
\item it receives at least two different proposals $\mathit{val}_i$ and $\mathit{val}_j$ such that $\mathit{val}_i$ is received from a replica sending its \textsc{Prepare} message, while $\mathit{val}_j$ is received from another replica sending its \textsc{Commit} message.
\end{enumerate*}

In order to grasp why the third case is optimal, consider the following observations:
\begin{enumerate}[label=(\arabic*), leftmargin=1.5em]
    \item If the Byzantine leader sends different proposals to a correct replica $p$, $p$ will detect the misbehavior of the leader and notify all replicas about it (line~\ref{line:notify:faultiness} in Algorithm~\ref{alg:probft}).
    Hence, sending multiple proposals to $p$ increases the probability that correct replicas avoid deciding some value(s).
    Note that there is no agreement violation when correct replicas avoid deciding a value.
    Hence, the Byzantine leader should send only \textit{one} proposal to each correct replica, i.e., $\Pi^i_C \cap \Pi^j_C =\emptyset$, for any $i,j \in \{1,\dots,m\}$, as it intends to increase the probability of agreement violation.
    Consequently, in the optimal case, the number of proposals the leader sends is bounded by the number of correct replicas, i.e., $m \leq n-f$.
    \item The Byzantine leader should send \textit{two} proposals to increase the probability of agreement violation.
    To show this result, we prove in Theorem~\ref{thm:m+1:to:m} that sending $m$ proposals instead of $m+1$ proposals, where $m\geq 2$, increases the probability of agreement violation.
    Given the previous observation that states the leader sends at most $n-f$ proposals, we can now say that the leader prefers to send $n-f-1$ proposals instead of sending $n-f$ proposals in the optimal case.
    Likewise, it prefers to send $n-f-2$ proposals instead of sending $n-f-1$ proposals.
    Following this line of reasoning, we conclude that the Byzantine leader should send two proposals to increase the probability of agreement violation.
    \item The probability of forming a probabilistic quorum by a replica $p$ for a proposal $\mathit{val}$ increases by expanding the set of replicas that send $\mathit{val}$.
    This result is formally presented in Theorem~\ref{cor:increase:t}.
    From the previous observation, we know that the leader should send two proposals.
    Now, we can say that the Byzantine leader should maximize the size of these two sets, resulting in the optimal case depicted in Figure~\ref{fig:3}.
\end{enumerate}

\begin{theorem}\label{thm:m+1:to:m}
Given a Byzantine leader who intends to send multiple proposals, consider the following two scenarios:
\begin{enumerate*}[label=(\arabic*)]
    \item given non-empty sets $\Pi^1,\dots,\Pi^{m+1}$ of replicas, where $m\geq 2$, and $|\Pi^1| \leq$
    $|\Pi^2|\leq \dots\leq |\Pi^{m+1}|$, the leader sends a distinct proposal to each set, and
    \item the leader merges two sets $\Pi^1$ and $\Pi^2$ to create a set $\Pi^{1,2}$ and sends $m$ proposals to $\Pi^{1,2},\Pi^3,\dots,\Pi^{m+1}$.
\end{enumerate*}
The probability of agreement violation in the second scenario is greater than in the first scenario.
\end{theorem}

\begin{theorem}\label{cor:increase:t}
Suppose any replica forms a quorum upon receiving $q$ messages. 
Consider a set of $r$ replicas, each of which takes a sample composed of $o \times q$ distinct replicas uniformly at random from $\Pi$, with the condition that $n < o\times r$, and sends a message to all sample members. 
The value of $r$ and the probability of a replica forming a quorum are directly proportional.
\end{theorem}

Since the third case discussed above is optimal, i.e., the probability of compromising the agreement is maximized by a Byzantine leader when it divides the correct replicas into two equally sized groups, $\Pi_C^1$ and $\Pi_C^2$, and sends a proposal $\mathit{val}_1$ to $\Pi_C^1 \cup \Pi_F$ and another proposal $\mathit{val}_2$ to $\Pi_C^2 \cup \Pi_F$, we only analyze this case.
We prove that given a Byzantine leader who may send several proposals, the probability of agreement violation in a view under the worst-case scenario is bounded by $\mathit{exp}\boldsymbol{(} -\Theta(\sqrt{n}) \boldsymbol{)}^4$ in the following theorem.

\begin{theorem}\label{thm:agreement:in:view}
Given a Byzantine leader who may send several proposals, the probability of agreement violation in a view under the worst-case scenario is at most $\mathit{exp}\boldsymbol{(} -\Theta(\sqrt{n}) \boldsymbol{)}^4$. 
\end{theorem}

\paragraph{Probabilistic Agreement with view change} 
We now consider the case of a view change.
When referring to agreement within different views, we need to guarantee that if at least one correct replica decides on a proposal $\mathit{val}$ in a view~$v$, the probability that some correct replica decides on a different proposal~$\mathit{val}'$ in a view $v' > v$ is negligible.
To guarantee this condition, we need to demonstrate that the leader of any view $v''>v$ proposes $\mathit{val}$ with high probability.

Recall that, in Algorithm \ref{alg:probft}, when the synchronizer notifies a replica to enter a new view $v''$, the replica informs the leader of~$v''$ about its latest prepared value through a $\textsc{NewLeader}$ message.  
The leader of~$v''$ waits until it observes a deterministic quorum of $\textsc{NewLeader}$ messages.
If at least $\lceil{(n+f+1)/2}\rceil$ correct replicas have prepared $\mathit{val}$, then the leader must propose $\mathit{val}$, regardless of its type, whether Byzantine or correct.
The problem occurs when $w < \lceil{(n+f+1)/2}\rceil$ correct replicas have prepared $\mathit{val}$.
Note that $w\geq 1$ as we assumed at least one correct replica has decided $\mathit{val}$.
One of the following scenarios can happen:

\begin{itemize}[label=--, leftmargin=1.1em]
    \item $1 \leq w \leq f$. 
    If the leader is Byzantine, it can propose any value different than $\mathit{val}$.
    Besides, if the leader is correct, it proposes a value $\mathit{val}' \neq \mathit{val}$ 
        %
    when the number of $\textsc{NewLeader}$ messages received from replicas that prepared $\mathit{val}'$ is greater than the number of $\textsc{NewLeader}$ messages received from replicas that prepared~$\mathit{val}$.
    \item $f+1 \leq w < \lceil{(n+f+1)/2}\rceil$. 
    If the leader is Byzantine, it can propose a value $\mathit{val}' \neq \mathit{val}$ if the number of replicas that prepared~$\mathit{val}'$ is greater than the number of replicas that prepared $\mathit{val}$.
    Besides, if the leader is correct, it proposes a value $\mathit{val}' \neq \mathit{val}$ when the number of $\textsc{NewLeader}$ messages received from replicas that prepared~$\mathit{val}'$ is greater than the number of $\textsc{NewLeader}$ messages received from replicas that prepared $\mathit{val}$.
\end{itemize}

We need to ensure that there is a high probability that the system will not be in these scenarios.
The following theorem shows this.

\begin{theorem}\label{thm:view:change}
    The probability of proposing a value~$\mathit{val}'$ when another value $\mathit{val}$ has been decided by a correct replica in a prior view is at most $\mathit{exp}\boldsymbol{(}-(q\times \delta^2)/((\delta+1)\times(\delta+2)) \boldsymbol{)} $, where $\delta = 2n/(o\times (n+f)) -1$.
\end{theorem}

\begin{corollary}[Main safety result]\label{thm:livenss:safety}
\probft guarantees safety with a probability of at least $1 - \mathit{exp}(-\Theta(\sqrt{n}))$.
\end{corollary}

\section{Numerical Evaluation}\label{sec:evaluation}
In this section, we illustrate the usefulness of \probft by presenting a numerical analysis of the agreement and termination probabilities, considering $q=2\sqrt n$ and $o \in \{1.6,1.7,1.8\}$. 

\paragraph{Analysis with fixed fault threshold and varying system sizes}
In Figure~\ref{fig:analysis}, the top sub-figures depict analyses with $f/n=0.2$ and varying system sizes.
The top-left sub-figure depicts the probability of ensuring agreement, considering the worst-case scenario in which there are faulty leaders in each view.
It shows that as the system size increases, the probability of ensuring agreement also increases.
Besides, the top-right sub-figure depicts the probability of terminating in a view after GST when the leader is correct.
It shows the probability of deciding increases as the number of replicas increases.

\paragraph{Analysis with fixed system size and a varying number of faulty replicas}
In Figure~\ref{fig:analysis}, the bottom sub-figures depict similar analyses with $n=100$ and a varying number of faulty replicas.
The bottom-left sub-figure shows that the probability of ensuring agreement increases as we have fewer Byzantine replicas.
Similarly, the bottom-right sub-figure shows that the probability of deciding increases as we have fewer Byzantine replicas.



\paragraph{Number of exchanged messages}
Figure~\ref{fig:pbft:probft:hotstuff} shows the number of exchanged messages in PBFT, \probft (for different values of $o$), and HotStuff.
In the figure, it is possible to see that \probft exchanges significantly fewer messages than PBFT despite having the same good-case optimal latency.
Taking together with the results of Figure~\ref{fig:analysis}, we can see that \probft with $o=1.7$ ensures a high probability of agreement and termination exchanging only $18$-$25\%$ of the messages required by PBFT.

\begin{figure}
    \centering
    \includegraphics[scale=0.755]{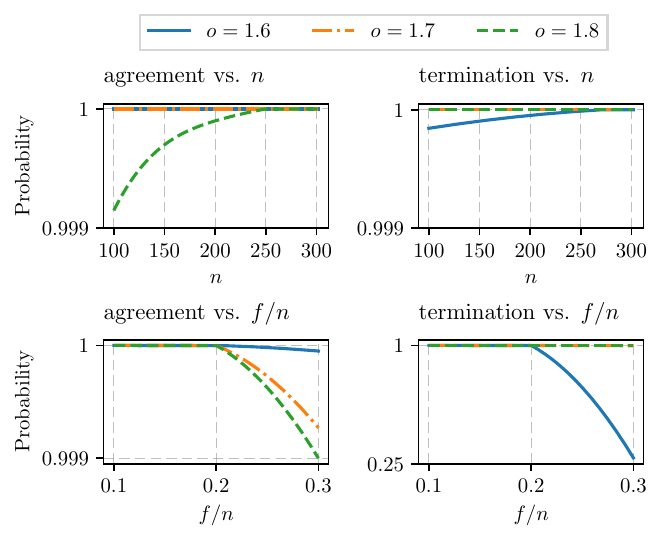}
    \caption{
    \probft agreement and termination probability analysis with $q=2\sqrt n$.
    For $f/n=0.2$, the top-left figure depicts the probability of ensuring agreement with faulty leaders in every view, 
    while the top-right figure shows the probability of terminating in a view after GST when the leader is correct.
    For $n=100$, the bottom-left figure depicts the probability of ensuring agreement with faulty leaders in every view, \
    while the bottom-right figure depicts the probability of terminating in a view after GST when the leader is correct.}
    \label{fig:analysis}
    \Description[]{}
\end{figure}

\section{Related Work}\label{sec:related:work}
\paragraph{Scalable BFT consensus protocols}
PBFT~\cite{pbft} is considered the baseline ``practical'' BFT protocol.
It is optimal in terms of resilience ($f < n/3$) and best-case latency (three communication steps) but employs an all-to-all message exchange pattern, which results in a quadratic message complexity, making it very costly in large deployments. 
With the advent of large-scale systems like blockchains and decentralized payment systems, the scalability of Byzantine consensus protocols has become a hot topic, with numerous contributions from both academia and industry.

There are several approaches to enhance the scalability of BFT consensus protocols. 
For instance, protocols like HotStuff~\cite{yin19hotstuff} modify the message exchange pattern from all-to-all to leader-to-all-to-leader, resulting in a linear message complexity.
Alternatively, to improve the message complexity and balance the load on the system, some protocols~\cite{kauri,li2020scalable} utilize a tree-based message exchange pattern, while other protocols (e.g.,~\cite{algorand,rapidchain,gosig,cason2021design}) employ a gossip layer for communication.
Similarly, there are some works that use expander graphs (and related techniques) to make Byzantine agreement more scalable (e.g.,~\cite{expandergraphs}).
Other works like RedBelly~\cite{redbelly} improve different aspects of the protocol, such as how ordered transactions are disseminated and verified, to achieve better scalability in practical deployments without reducing the message complexity of the base protocol.
While most of these works avoid the all-to-all communication pattern, they increase the number of communication steps necessary to achieve consensus.

\paragraph{Randomized consensus protocols}
The primary drive for developing randomized consensus protocols was the well-known FLP impossibility result~\cite{flp}. 
This result states that for crash-prone asynchronous systems, designing a deterministic consensus protocol is impossible.
One of the main approaches to circumvent such an impossibility result involves relaxing deterministic termination to probabilistic termination, assuming processes have access to random numbers.

The randomized consensus protocol by Ben-Or~\cite{ben-or} assumes a strong adversary who can observe the entire history of the system and uses a local coin. 
The protocol operates in rounds, with each round involving $O(n^2)$ message exchanges, and it requires exponential expected time to converge in the worst case.
Subsequent works like RITAS~\cite{ritas} and WaterBear~\cite{waterbear} showed this type of protocol can be made practical.


Rabin~\cite{rabin1983randomized} showed that the same type of protocol could achieve termination in the expected constant time by resorting to a common coin.
The construction of this common coin typically requires strong cryptographic primitives, which can negatively affect the performance of the protocol. 
Further works like Cachin et al.~\cite{cachin2000random} leveraged threshold signatures to devise a protocol with constant expected time and message complexity of $O(n^2)$.
More recently, Mostéfaoui et al.~\cite{mostefaoui2015signature} proposed a similar protocol that is signature-free.
HoneyBadgerBFT~\cite{miller2016honey} is a practical randomized consensus protocol that has at least an $O(n^3)$ message complexity.
Follow-up works (e.g.,~\cite{duan2018beat,dumbong,Liu23}) improved different aspects of practical randomized protocols, including their message complexity, but never below $O(n^2)$.
The same can be said about DAG-based protocols (e.g.,~\cite{dagrider,bullshark}) that use local rules to commit blocks (i.e., decide consensus values) and resort to a kind of randomized consensus only when such rules are not enough.


\paragraph{Synchronizer}
Byzantine fault-tolerant consensus protocols designed for partially synchronous systems typically structure their execution in a sequence of views, with the premise that there will be a view in which all correct replicas will overlap with enough time to reach a consensus if there is a correct leader. 
Designing these protocols is challenging, and researchers usually pay more attention to guaranteeing the system's safety rather than liveness~\cite{pbft,sbft,yin19hotstuff}.
The problem with the partially synchronous model and designing the protocols in a sequence of views is that replicas may diverge indefinitely before the GST, reaching GST in different views. 
This problem is typically not addressed in commonly used Byzantine fault-tolerant protocols~\cite{pbft,bessani2014state,Naor2021Cogsworth} in which the liveness is argued based on the assumption that after the system reaches GST, all the correct replicas will eventually converge to the same view.
By separating the mechanism used for view synchronization in a distinct component, Bravo~et~al.~\cite{makingByzConsLive} formally defined the synchronizer abstraction, which we employ in \probft.
Notice that using such abstraction does not incur added message complexity as there are constructions with linear message complexity~\cite{linearsync,fever}. 

\paragraph{Probabilistic quorum systems}
Malkhi et al.~\cite{malkhi01probabilistic} introduced probabilistic quorum systems to enhance the efficiency of data replication by relaxing strict quorum requirements and allowing for probabilistic guarantees of consistency.
These quorum systems operate under the implicit assumption that any chosen quorum will be accessible without taking into account the potential effects of failures or asynchrony~\cite{availability_non_strictQS,originQS,signedQS}.
In other words, it does not account for the impact of an adversarial scheduler (also known as an active adversary~\cite{availability_non_strictQS} or an asynchronous scheduler~\cite{signedQS}) that could potentially delay the delivery of messages.
Yu~\cite{signedQS} introduced an alternative concept termed signed quorum systems, aiming to address the challenges posed by network scheduling. 
Nevertheless, Yu's method remains susceptible to manipulation by an adversarial scheduler~\cite{availability_non_strictQS}.

\section{Conclusion}\label{sec:conclusion}
We presented \probft, a Byzantine fault-tolerant consensus protocol that ensures safety and liveness with high probability in permissioned partially synchronous systems.
This protocol's message complexity is $O(n\sqrt{n})$ in a system with $n$ replicas, and it has an optimal number of communication steps.
\probft introduces a novel paradigm for designing scalable Byzantine-resilient protocols for less pessimistic contexts. 
We believe the same techniques used in \probft can be employed in other types of quorum-based protocols. 
As future work, we are particularly interested in leveraging \probft for constructing a scalable state machine replication protocol and a streamlined blockchain consensus, eliminating the need for a view-change sub-protocol.

\begin{acks}
This work was partially supported by the \grantsponsor{dfg}{Deutsche Forschungsgemeinschaft (DFG, German Research Foundation)}{https://www.dfg.de/} -- \grantnum{dfg}{446811880 (BFT2Chain)}, 
and by \grantsponsor{fct}{FCT} \ \ through the    
\grantnum{smartchain}{\href{https://doi.org/10.54499/2022.08431.PTDC}{SMaRtChain project (2022. 08431.PTDC)}} and the
\grantnum{lasige}{\href{https://doi.org/10.54499/UIDB/00408/2020}{LASIGE Research Unit (UIDB/00408/2020} and \href{https://doi.org/10.54499/UIDP/00408/2020}{UIDP/00408/2020)}}.
\end{acks}

\ifbool{extendedVersion}{}{\balance}
\bibliographystyle{ACM-Reference-Format}
\bibliography{ref.bib}

\clearpage
\appendix

\section{Probability Definitions and Bounds}
We utilize several probability definitions and bounds to establish the correctness of \probft. 
This appendix presents such definitions and bounds.

\subsection*{Chernoff bounds} 
We use the Chernoff bounds~\cite{motwani1995randomized} for bounding the probability that the sum of independent random variables deviates significantly from its expected value.
Suppose $X_1, \dots, X_n$ are independent Bernoulli random variables, and let $X$ denote their sum.
Then, for any $\delta\in(0,1)$:
\begin{align}
&\Pr\boldsymbol{(} X \leq (1-\delta)\mathop{\mathbb{E}}[X] \boldsymbol{)} \leq \mathit{exp}( -\delta^2\mathop{\mathbb{E}}[X]/2 ).\label{ineq:chernof2}
\end{align}
Besides, for any $\delta \geq 0$:
\begin{align}
\Pr\boldsymbol{(} X \geq (1+\delta)\mathop{\mathbb{E}}[X] \boldsymbol{)} \leq \mathit{exp}\boldsymbol{(} -\delta^2\mathop{\mathbb{E}}[X]/(2+\delta) \boldsymbol{)}.\label{ineq:chernof}
\end{align}

\subsection*{Negative association} 
In randomized algorithms and analysis, using independent random variables is common, enabling the application of powerful theorems and bounds like Inequalities~\ref{ineq:chernof2} and \ref{ineq:chernof}. 
However, random variables may not always be independent.
The following definition and theorem allow us to leverage Inequalities~\ref{ineq:chernof2} and \ref{ineq:chernof} when random variables are dependent.

\begin{definition}[Negative association~\cite{concentration}]
The random variables $X_i, i\in \{1,\dots, n\}$, are negatively associated if for all disjoint subsets $I,J\subseteq \{1,\dots, n\}$ and all non-decreasing functions $f$ and $g$,
\begin{align*}
\mathop{\mathbb{E}}[f(X_i, i\in I)g(X_j,j\in J)] \leq \mathop{\mathbb{E}}[f(X_i, i\in I)]\mathop{\mathbb{E}}[g(X_j,j\in J)].
\end{align*}
\end{definition}

\begin{theorem}[Chernoff–Hoeffding bounds with negative dependence~\cite{concentration}]\label{thm:chernoff:NA}
The Chernoff–Hoeffding bounds can be applied to $X=\sum_{i\in \{1,\dots, n\}}X_i$ if the random variables $X_1,\dots,X_n$ are negatively associated.
\end{theorem}

In order to ease showing negative association, two following properties are considered~\cite{concentration}:
\begin{itemize}[label=--, leftmargin=1em]
\item \textbf{Closure under Products:} If $X_1, \ldots, X_n$ and $Y_1, \ldots, Y_m$ are two independent families of random variables that are separately negatively associated, then the family $X_1, \ldots, X_n,$ $Y_1, \ldots, Y_m$ is also negatively associated.

\item \textbf{Disjoint Monotone Aggregation:} If $X_i$, $i \in \{1,\dots, n\}$, are negatively associated and $\mathcal{A}$ is a family of disjoint subsets of $\{1,\dots, n\}$, then the random variables $f_A(X_i, i \in A)$, $A \in \mathcal{A}$, are also negatively associated, where the $f_A$s are arbitrary non-decreasing (or non-increasing) functions.
\end{itemize}

We also use the following result adapted from~\cite{narv} in the process of showing a negative association.
\begin{theorem}\label{thm:na:rv}
    Consider $n$ distinct items from which a sample of size $s$ is chosen without replacement.
    Let $X_i, i\in \{1,\dots, n\}$, be random variables indicating the presence of a specific item in the sample.
    Random variables $X_i, i\in \{1,\dots, n\}$, are negatively associated.
\end{theorem}

\subsection*{Hypergeometric distribution} 
The hypergeometric distribution, denoted as $\mathcal{HG}(N, M, r)$, characterizes the number of specific items within a random sample of size $r$, drawn without replacement from a population of size~$N$, containing $M$ items of the same type~\cite{probBook}. 
The expected value of a random variable $X$ following this distribution equals $rM/N$. 
Additionally, a tail bound for $X$ can be derived as follows:
\begin{equation}\label{ieq:tail:bound}
    \Pr\big( X \leq  \mathop{\mathbb{E}}[X] - rt \big) \leq \mathit{exp}(-2rt^2),
\end{equation}
with $t$ taking values within the interval $(0, M/N)$~\cite{CHVATAL1979285,skala2013hypergeometric}.

\section{Probability of Forming a Probabilistic Quorum}
Here we present preliminary results that will be used to analyze \probft.
Recall that in \probft, in each view~$v$, upon receiving the first proposal from the leader of~$v$, a correct replica takes a random sample of size~$o \times q$ and sends a \textsc{Prepare} message to all members of that sample.
Upon receiving~$q$ valid \textsc{Prepare} messages, a correct replica forms a probabilistic quorum.
Here we compute the probability that a correct replica forms a probabilistic quorum when all correct replicas send their \textsc{Prepare} messages.
We later use this result to compute the probability of ensuring termination after GST and when the leader is correct.
With this aim, given a correct process~$i$, assuming that $r$ processes send their \textsc{Prepare} messages, we first compute the expected number of processes that have~$i$ in their samples.
Using such an expected value and the Chernoff bound, we then provide a lower bound for the probability that~$i$ forms a probabilistic quorum.
Last, we set~$r=|\Pi_C|$ and present the desired probability.

\begin{lemma}\label{lem:expect}
Let $R$ be a subset of replicas with a size of $r$, where each replica $i \in R$ takes a sample composed of $s$ distinct replicas uniformly at random from~$\Pi$.
The expected number of replicas in~$R$ with a given replica~$j$ in their samples is $r \times s / n$.
\end{lemma}
\begin{proof}
For each replica $i \in R$, we define an indicator random variable $I_{ij}$ identifying whether replica~$j$ is in the sample of $i$:  
\begin{align*}
    I_{ij} = 
    \begin{cases}
           1 \quad \text{$j$ is in the sample of $i$}
        \\ 0 \quad \text{otherwise.}
    \end{cases}
\end{align*}
The expected value of $I_{ij}$ can be computed as follows:
\begin{align*}
    \mathop{\mathbb{E}}[I_{ij}] 
    = \Pr( \text{$j$ is in the sample of $i$} )
    = s / n.
\end{align*}
Let $I_j = \sum_{i\in R}I_{ij}$, indicating the number of replicas in~$R$ that have $j$ in their samples.
Using the linearity of expectation, the expected value of $I_{j}$ can be computed as follows:
\begin{align*}
     \mathop{\mathbb{E}}[I_j] = \sum_{i\in R}\mathop{\mathbb{E}}[I_{ij}] = r\times s / n.
\end{align*}
\end{proof}

\begin{theorem}\label{thm:1}
Let $R$ be a subset of replicas with size $r$, where each replica $i \in R$ takes a sample composed of $s = o \times q$ distinct replicas uniformly at random from $\Pi$.
If $n < o \times r$, then at least $q$ replicas within~$R$ have a given replica~$j$ in their samples with a probability of at least 
$1 - \mathit{exp}\boldsymbol{(}- ((s\times r)/(2n)) \times (1-(n/(o \times r)))^2 \boldsymbol{)}$.
\end{theorem}
\begin{proof}
Let $I_j$ be defined in the same way as in Lemma~\ref{lem:expect}, representing the number of replicas in~$R$ that have replica~$j$ in their samples.
We need to find a lower bound for $\Pr( I_j \geq q )$.
Since $I_j$ is obtained by taking the summation of $r$ i.i.d. indicator random variables, we can use the Chernoff bound~\ref{ineq:chernof2} to find a lower bound for $\Pr( I_j \geq q )$.
If $n < o \times r$, we have:
\begin{align*}
\Pr( I_j \geq q ) 
& = 1 - \Pr( I_j \leq q ) 
\\& = 1 - \Pr( I_j \leq s/o)  \qquad\qquad\qquad\qquad\qquad\text{(since $s=o \times q$)}
\\& = 1 - \Pr\left( I_j \leq \frac{s\times r}{n} \times \frac{n}{o\times r} \right) 
\\& = 1 - \Pr\left( I_j \leq  \frac{s\times r}{n} \times \boldsymbol{(}1-(1-\frac{n}{o\times r})\boldsymbol{)} \right) 
\\& \geq 1 - \mathit{exp}\left(- \frac{\frac{s\times r}{n} \times \left(1-\frac{n}{o\times r}\right)^2 }{2} \right). 
\end{align*}
As $\mathop{\mathbb{E}}[I_j] = r\times s / n$ according to Lemma~\ref{lem:expect}, the last line holds using the Chernoff bound~\ref{ineq:chernof2} by assuming $\delta = 1-n/(o\times r)$.
Note that in Chernoff bound~\ref{ineq:chernof2}, $\delta$ needs to be in the interval $(0,1)$.
Hence, we need to assume that $n < o\times r$.
\end{proof}

\begin{corollary}\label{thm:s:oq:q}
Suppose each correct replica takes a sample composed of $s = o \times q$ distinct replicas uniformly at random from $\Pi$ and multicasts a message to all members of the sample.
Provided that a replica forms a probabilistic quorum upon receiving $q$ messages if $n < o\times(n-f)$, the probability of forming a probabilistic quorum by a replica is at least $1 - \mathit{exp}\boldsymbol{(}- q(c-1)^2/(2c) \boldsymbol{)}$, where $c=o\times (n-f)/n$.  
\end{corollary}
\begin{proof}
Given a replica $j$, let $I_j$ be defined in the same way as in Lemma~\ref{lem:expect}, representing the number of correct replicas that have $j$ in their samples.
We have:
\begin{align*}
\Pr\left( \text{$j$ forms a probabilistic quorum} \right) = \Pr\left( I_j \geq q \right).    
\end{align*} 
We now use Theorem~\ref{thm:1} to find a lower bound for $\Pr\left( I_j \geq q \right)$.
Note that in Theorem~\ref{thm:1}, $r$ is the number of replicas that take random samples.
Since all correct replicas take random samples, $r=n-f$.
We have:
\begin{align*}
\Pr\left( I_j \geq q \right) 
& \geq 1 - \mathit{exp}\left(- \frac{\frac{s\times r}{n} \times \left(1-\frac{n}{o\times r}\right)^2 }{2} \right) &\hfill \text{(from Theorem~\ref{thm:1})}
\\& = 1 - \mathit{exp}\left(- \frac{\frac{o\times q(n-f)}{n} \times \left(1-\frac{n}{o\times(n-f)}\right)^2}{2} \right) &\hfill \text{(since $s=o\times q$, and $r=n-f$)}
\\& = 1 - \mathit{exp}\left(- \frac{q(c-1)^2}{2c} \right),
\end{align*} 
where $c = o\times(n-f)/n$.
\end{proof}

\section{Theorems related to the optimal behavior of a Byzantine leader}
This section presents two theorems related to the optimal behavior of a Byzantine leader.
In the first theorem, we compute a relationship between the number of replicas that send \textsc{Prepare} (resp. \textsc{Commit}) messages and the probability of forming a probabilistic quorum in the \emph{prepare} (resp. \emph{commit}) phase.

\begin{theorem}[Theorem~\ref{cor:increase:t}]
Suppose a replica forms a probabilistic quorum upon receiving~$q$ messages. 
Consider~$r$ replicas, each of which takes a sample composed of $s = o \times q$ distinct replicas uniformly at random from $\Pi$, with the condition that $n < o\times r$, and sends a message to all members of the sample. 
The value of $r$ and the probability that a replica forms a probabilistic quorum are directly proportional.
\end{theorem}
\begin{proof}
Consider two systems where the first comprises a finite set $\Pi$, and the second comprises a finite set $\Pi'$, with $|\Pi| = |\Pi'| = n$.
Further, consider a set $R \subseteq \Pi$ of size $r$ and another set $R' \subseteq \Pi'$ of size $r'$ such that $r<r'$, $n < o \times r$, and $n < o\times r'$.
Suppose any replica within $R$ (resp. $R'$) takes a sample of size~$s$ uniformly at random from $\Pi$ (resp. $\Pi'$) and sends a message to all members of the sample.
Let $i \in R$ and $j \in R'$.
In order to prove the theorem, we need to show that the probability of~$j$ forming a probabilistic quorum is greater than the probability of $i$ forming a probabilistic quorum.
If a random variable~$X$ (resp. $Y$) denotes the number of replicas that have $i$ (resp. $j$) in their samples, we need to show that $\Pr(Y \geq q) > \Pr(X \geq q)$, or equivalently, $\Pr(Y \leq q) < \Pr(X \leq q)$.

As any replica within $R$ (resp. $R'$) takes its sample without replacement, the probability of replica~$i$ (resp. $j$) being within the sample equals $s/n$.
It is clear that $X \sim \mathit{Bin}(r,s/n)$ and $Y \sim \mathit{Bin}(r',s/n)$.
Thus,
\begin{align*}
  & \Pr(X \leq q) = \sum_{k = 0}^{q} \binom{r }{k}(s/n)^k\big(1-(s/n)\big)^{r -k},
\\& \Pr(Y \leq q) = \sum_{k = 0}^{q} \binom{r'}{k}(s/n)^k\big(1-(s/n)\big)^{r'-k}.
\end{align*}
Without loss of generality, assume that $r'=r+1$.
In order to show $\Pr(Y \leq q) \leq \Pr(X \leq q)$, we show that for each $k \in [0,q]$,
\begin{align*}
    \binom{r+1}{k}(s/n)^k\big(1-(s/n)\big)^{r+1-k} < \binom{r }{k}(s/n)^k\big(1-(s/n)\big)^{r -k}.
\end{align*}
We have:
\begin{align*}
              & n < o\times r' = o\times(r+1)                      &\hfill (\text{by assumption})
\\& \Rightarrow q < (r+1)(o \times q/n)
\\& \Rightarrow q < (r+1)(s/n)                                     &\hfill (s=o\times q)
\\& \Rightarrow q < (r+1)\big(1-\boldsymbol{(}1-(s/n)\boldsymbol{)}\big)
\\& \Rightarrow (r+1)\boldsymbol{(}1-(s/n)\boldsymbol{)} < r+1-q
\\& \Rightarrow (r+1)\boldsymbol{(}1-(s/n)\boldsymbol{)} < r+1-k   &\hfill (\text{since } r+1-q\leq r+1-k)
\\& \Rightarrow \frac{r+1}{r+1-k} < \boldsymbol{(}1-(s/n)\boldsymbol{)}^{-1}
\\& \Rightarrow \frac{r+1}{r+1-k} < \boldsymbol{(}1-(s/n)\boldsymbol{)}^{r-k-(r+1-k)}
\\& \Rightarrow \binom{r+1}{k}(s/n)^k\boldsymbol{(}1-(s/n)\boldsymbol{)}^{r+1-k} < \binom{r }{k}(s/n)^k\boldsymbol{(}1-(s/n)\boldsymbol{)}^{r -k}.
\end{align*}
\end{proof}

\begin{lemma}\label{lem:new:lem}
    Given a Byzantine leader who intends to maximize the probability of disagreement in a view,
    the probability that correct replicas detect the faultiness of the leader is inversely proportional to the probability of having disagreement. 
\end{lemma}
\begin{proof}
    Assume that the leader~$l$ of view~$v$ is Byzantine and intends to maximize the probability of disagreement.
    Recall that \probft provides a mechanism through which correct replicas can detect the faultiness of~$l$ (lines~\ref{line:block_view_start}-\ref{line:notify:faultiness}).
    After detecting the faultiness of~$l$, a correct replica informs other replicas and waits for a view-change.
    Indeed, a requirement for deciding a value by a correct replica~$i$ in view~$v$ is that~$i$ must not have detected the faultiness of~$l$.
    Accordingly, if the probability of detecting the faultiness of~$l$ by~$i$ increases (resp. decreases), then the probability of deciding a value by~$i$ decreases (resp. increases).
    Note that if~$i$ does not decide on a value in view~$v$, it does not contribute to creating a disagreement.
    It follows that if the probability of detecting the faultiness of~$l$ by~$i$ increases (resp. decreases), the probability that~$i$ contributes to creating a disagreement decreases (resp. increases).
\end{proof}

\begin{theorem}[Theorem~\ref{thm:m+1:to:m}]
Given a Byzantine leader who intends to send multiple proposals, consider the following two scenarios:
\begin{enumerate*}[label=(\arabic*)]
    \item the leader creates non-empty sets $\Pi^1,\dots,\Pi^{m+1}$ of replicas, where $m\geq 2$ and $|\Pi^1|\leq |\Pi^2|\leq \dots\leq |\Pi^{m+1}|$, and sends a distinct proposal to each set, and
    \item the leader merges two sets $\Pi^1$ and $\Pi^2$ to create a set $\Pi^{1,2}$ and sends $m$ proposals to $\Pi^{1,2},\Pi^3,\dots,\Pi^{m+1}$.
\end{enumerate*}
The probability of agreement violation in the second scenario is greater than in the first scenario.
\end{theorem}
\begin{proof}
Assume that the leader~$l$ of view~$v$ is Byzantine and intends to maximize the probability of disagreement.
We need to show that the probability of agreement violation in the second scenario is greater than in the first scenario.
From Lemma~\ref{lem:new:lem}, the probability of agreement violation is inversely proportional to the probability of detecting the faultiness of~$l$.
Accordingly, the proof is complete by showing that the probability of detecting the faultiness of~$l$ in the second scenario is less than in the first scenario.
For this purpose, consider a correct replica~$i$.
Assume that $i \in \Pi^k$.
We can consider the following cases for the values of $k$:
\begin{enumerate}[label=(\arabic*),leftmargin=1.5em]
    \item $3\leq k \leq m+1$. 
        The probability of finding out the faultiness of the leader by $i$ does not change by altering the scenarios, as from the viewpoint of $i$, the number of replicas whose proposal is different is not changed. 
    \item $1\leq k \leq 2$.
        The probability of finding out the faultiness of the leader by $i$ in the second scenario is less than in the first scenario.
        This is because there are more replicas that propose the same value in the second scenario, and according to Theorem~\ref{cor:increase:t}, the probability of deciding a proposal by~$i$ increases in the second scenario.
        As the probability of deciding a proposal increases, the probability that~$i$ finds out the faultiness of the leader decreases.
\end{enumerate}
According to these cases, the probability of finding out the faultiness of the leader by~$i$ in the second scenario is less than in the first scenario, completing the proof.
\end{proof}

\section{Analysis of \probft}\label{appendix:probft}
\subsection{The Probability of Termination}\label{appendix:termination}
To ensure that \probft satisfies Probabilistic Termination, we need to guarantee that, after GST, every correct replica terminates in a view $v$ with high probability if the leader of $v$ is correct.
As described in Section~\ref{sec:probft}, in \probft, a correct replica decides on a value if it has observed a probabilistic quorum of $\textsc{Prepare}$ messages and a probabilistic quorum of $\textsc{Commit}$ messages. 
However, \probft lacks deterministic quorums compared to PBFT, so it is essential to guarantee that the probability of observing probabilistic quorums is high. 

Notice that the number of replicas sending $\textsc{Commit}$ messages will be lower than the number of replicas sending \textsc{Prepare} messages. 
However, in Theorem~\ref{cor:epsilon}, we stated that a replica observes a quorum with a high probability if the leader is correct so that most replicas will observe quorums and can terminate. 

\begin{theorem}[Theorem~\ref{cor:epsilon}]
Suppose each correct replica takes a sample composed of $o \times l\sqrt{n}$ distinct replicas, uniformly at random from $\Pi$, and multicasts a message to them, where $ o \in [1, (n/(n-f))\times (2+\sqrt{3})] $, and $l\geq 1$.
Provided that a replica forms a quorum upon receiving $q = l\sqrt{n}$ messages, the probability of forming a quorum for a replica is at least $1 - \mathit{exp}(- \sqrt{n} )$.
\end{theorem}
\begin{proof}   
From Corollary~\ref{thm:s:oq:q}, if $q = l\sqrt{n}$, then a replica forms a quorum upon receiving $q$ messages with a probability of at least $1 - \mathit{exp}\big(- \boldsymbol{(}l\sqrt{n}\times(c-1)^2\boldsymbol{)} / (2c) \big)$, where $c=o\times (n-f)/n$.
Note that if $l = (2c) / (c-1)^2$, the probability is at least $1 - \mathit{exp}(- \sqrt{n} )$.
Based on the assumption that $l\geq 1$, it follows that $c$ must be in the interval $[2-\sqrt{3},2+\sqrt{3}]$.
Accordingly, $2-\sqrt{3} \leq o\times (n-f)/n \leq 2+\sqrt{3}$.
Thus, $(2-\sqrt{3})\times(n/(n-f)) \leq o \leq (2+\sqrt{3})\times(n/(n-f))$.
\end{proof}

\begin{lemma}\label{lem:commit:alpha}
If the leader is correct, then a correct replica receives $\textsc{Commit}$ messages from at least $q$ correct replicas with a probability of at least $1 - \mathit{exp}\boldsymbol{(} -(\alpha-q)^2/(2\alpha) \boldsymbol{)}$, where $\alpha = (s/n)\times(n-f)\times\boldsymbol{(}1-\mathit{exp}(-\sqrt{n})\boldsymbol{)}$.
\end{lemma}
\begin{proof}
Suppose $i$ is a correct replica for which we aim to compute the probability of receiving $\textsc{Commit}$ messages from $q$ correct replicas.
For each correct replica $j$, let $X_j$ denote a random variable defined as follows:
\begin{align*}
X_j = 
\begin{cases}
   1 & \text{if $j$ sends a $\textsc{Commit}$ message to $i$}
\\ 0 & \text{otherwise}.
\end{cases}
\end{align*}
For simplicity of notation, we assume that $\Pi_C=\{1,\dots,n-f\}$ in this proof. 
Furthermore, let $X = \sum_{j=1}^{n-f}X_j$ indicating the number of correct replicas that send $\textsc{Commit}$ messages to $i$.
Note that $X \geq q$ means $i$ receives $\textsc{Commit}$ messages from at least $q$ correct replicas.
Consequently, we need to compute a lower bound for $\Pr(X \geq q)$.

It is essential to note that random variables $X_j, j\in\Pi_C$ are not independent due to the following reason.
Each random variable $X_j, j\in\Pi_C$ equals 1 if $j$ sends a $\textsc{Commit}$ message to $i$.
Replica $j$ sends a $\textsc{Commit}$ message to $i$ if 
(a) it forms a quorum of $\textsc{Prepare}$ messages
and (b) $i$ is in the random sample of $j$. 
The point is that the probability of forming a quorum of $\textsc{Prepare}$ messages by~$j$ depends on the probability of forming a quorum of $\textsc{Prepare}$ messages by any other replica, as sampling in \probft is done without replacement.

As random variables $X_j, j\in\Pi_C$ are dependent, we cannot directly use the Chernoff bound~\ref{ineq:chernof2} to compute a lower bound for $\Pr(X \geq q)$.
However, we will show that random variables $X_j, j\in\Pi_C$ are negatively associated, so the Chernoff bound can be employed to compute a lower bound for $\Pr(X \geq q)$ according to Theorem~\ref{thm:chernoff:NA}.
To do so, given a correct replica $k$, for each correct replica $l \in \Pi_C$, let random variable $X_{k,l}$ denote the event of sending a $\textsc{Prepare}$ message from $k$ to $l$.
According to Theorem~\ref{thm:na:rv}, random variables $X_{k,l}, l \in \Pi_C$ are negatively associated.
Based on the closure under products property, 
\begin{equation*}
    X_{1,1},\dots,X_{1,n-f}, X_{2,1},\dots,X_{2,n-f}, \dots, X_{n-f,1},\dots,X_{n-f,n-f}
\end{equation*}
are negatively associated.
We now define a non-decreasing function $I$ as follows:
\begin{align*}
I(l) = 
\begin{cases}
   1 & \text{if } \sum_{k=1}^{n-f}X_{k,l} \geq q
\\ 0 & \text{otherwise.}
\end{cases}
\end{align*}
Note that function $I(l)$ for each $l\in \Pi_C$ is indeed an indicator random variable determining whether~$l$ forms a quorum of $\textsc{Prepare}$ messages.
According to the disjoint monotone aggregation property, random variables $I(l), l\in \Pi_C$ are negatively associated.

Notice that a correct replica $j$ sends a $\textsc{Commit}$ message to $i$ if and only if (a) $j$ forms a quorum of $\textsc{Prepare}$ messages and (b) $i$ is in the random sample of $j$.
As $i$ being in the random sample of $j$ is independent of forming a quorum of $\textsc{Prepare}$ messages by $j$, as well as being independent of $i$ being in the random sample of any other correct replica, random variables $X_j, j\in\Pi_C$ are negatively associated.
We are now ready to compute a lower bound for $\Pr(X \geq q)$.
With this aim, let us first compute the expected value of $X$.
Recall that $X = \sum_{j=1}^{n-f}X_j$; using the linearity of expectation, we have:
\begin{align*}
\mathop{\mathbb{E}}[X] 
  & = \sum_{j=1}^{n-f}\mathop{\mathbb{E}}[X_j]
\\& = \sum_{j=1}^{n-f}\Pr(X_j = 1)
\\& = \sum_{j=1}^{n-f}\Pr\big( I(j) = 1 \land \text{$i$ is in the random sample of $j$} \big)
\\& = \sum_{j=1}^{n-f}\Pr\big( I(j) = 1\big) \times \Pr(\text{$i$ is in the random sample of $j$} )
\\& \geq \sum_{j=1}^{n-f}(1-\mathit{exp}(-\sqrt{n})) \times (s/n) \qquad \text{(using Theorem~\ref{cor:epsilon})}
\\& = (s/n)\times(n-f)\times\big(1-\mathit{exp}(-\sqrt{n})\big).
\end{align*}
By assuming $\delta = 1- (q/\mathop{\mathbb{E}}[X])$, we now use the Chernoff bound~\ref{ineq:chernof2} to compute a lower bound for the desired probability:
\begin{align*}
\Pr(X \geq q)
  & = 1 - \Pr(X \leq q)
\\& = 1 - \Pr\big(X \leq (1-\delta)\mathop{\mathbb{E}}[X]\big)
\\& \geq 1 - \mathit{exp}( -\delta^2\mathop{\mathbb{E}}[X]/2 )
\\& = 1 - \mathit{exp}\big( -(\mathop{\mathbb{E}}[X]-q)^2/(2\mathop{\mathbb{E}}[X]) \big)
\\& \geq 1 - \mathit{exp}\big( -(\alpha-q)^2/(2\alpha) \big),
\end{align*}
where $\alpha = (s/n)\times (n-f)\times \boldsymbol{(}1-\mathit{exp}(-\sqrt{n})\boldsymbol{)}$.
\end{proof}

\begin{lemma}\label{lem:termination}
If the leader is correct, then a correct replica terminates, i.e., it decides a value, with a probability of at least $1 - \mathit{exp}\boldsymbol{(} -(\alpha-q)^2/(2\alpha) \boldsymbol{)} - \mathit{exp}(-\sqrt{n})$, where $\alpha = (s/n)\times(n-f)\times \boldsymbol{(}1-\mathit{exp}(-\sqrt{n})\boldsymbol{)}$.
\end{lemma}
\begin{proof}
Suppose the leader is correct.
A correct replica $i$ terminates by happening the following two events:
\begin{itemize}[label=--, leftmargin=1em]
\item $E_p$: $i$ forms a quorum of $\textsc{Prepare}$ messages,
\item $E_c$: $i$ forms a quorum of $\textsc{Commit}$ messages.
\end{itemize}
Note that these two events are dependent. 
Thus,
\begin{align*}
\Pr(\text{$i$ terminates}) 
& = \Pr(E_c \cap E_p)
\\& = 1 - \Pr(\bar{E}_c \cup \bar{E}_p)
\\& \geq 1 - \Pr(\bar{E}_c) - \Pr(\bar{E}_p)
\\& = \Pr(E_c) + \Pr(E_p) - 1
\\& = 1 - \mathit{exp}\big( -(\alpha-q)^2/(2\alpha) \big) - \mathit{exp}(-\sqrt{n}),
\end{align*}
where $\alpha = (s/n)\times(n-f)\times\boldsymbol{(}1-\mathit{exp}(-\sqrt{n})\boldsymbol{)}$.
The last line holds due to Theorem~\ref{cor:epsilon} and Lemma~\ref{lem:commit:alpha}.
\end{proof}

\begin{theorem}\label{thm:termination}
If the leader is correct, then every correct replica terminates, i.e., it decides a value, with a probability of at least $1 - (n-f)\left(\mathit{exp}( -(\alpha-q)^2/(2\alpha) ) - \mathit{exp}(-\sqrt{n})\right)$, where $\alpha = (s/n)(n-f)(1-\mathit{exp}(-\sqrt{n}))$.
\end{theorem}
\begin{proof}
For each replica $i\in \Pi_C$, let $E_i$ be the event deciding a value by $i$.
Using Lemma~\ref{lem:termination}, we have:
\begin{align*}
& \Pr(\cap_{i\in \Pi_C}E_i) 
\\&\quad = 1 - \Pr(\cup_{i\in \Pi_C}\bar{E}_i)
\\&\quad \geq 1 - \sum_{i\in \Pi_C}\Pr(\bar{E}_i)
\\&\quad \geq 1 - \sum_{i\in \Pi_C}\big(1-\Pr(E_i)\big)
\\&\quad \geq 1 - (n-f) + \sum_{i\in \Pi_C}\Pr(E_i) 
\\&\quad \geq 1 - (n-f) + (n-f)\Big(1 - \mathit{exp}( -(\alpha-q)^2/(2\alpha) ) - \mathit{exp}(-\sqrt{n})\Big)
\\&\quad \geq 1 - (n-f)\left(\mathit{exp}( -(\alpha-q)^2/(2\alpha) ) - \mathit{exp}(-\sqrt{n})\right),
\end{align*}
where $\alpha = (s/n)\times(n-f)\times\boldsymbol{(}1-\mathit{exp}(-\sqrt{n})\boldsymbol{)}$.
\end{proof}

\begin{theorem}[Theorem~\ref{cor:termination}]
After GST, if the leader of view $v$ is correct, then every correct replica decides a value in view $v$ with a probability of at least $1 - 2(n-f)\times\mathit{exp}\boldsymbol{(}-\Theta(\sqrt{n})\boldsymbol{)}$.
\end{theorem}
\begin{proof}
According to Theorem~\ref{thm:termination}, every correct replica terminates, i.e., it decides a value, with a probability of at least $1 - (n-f)\times\left(\mathit{exp}( -(\alpha-q)^2/(2\alpha) ) - \mathit{exp}(-\sqrt{n})\right)$, where $\alpha = (s/n)\times(n-f)\times(1-\mathit{exp}(-\sqrt{n}))$.
As $f = [0, n/3)$, we have:
\begin{align*}
\alpha > (2s/3)(1-\mathit{exp}(-\sqrt{n})) \approx 2s/3.  
\end{align*}
Since $s=o \times q$, and $q \approx \sqrt{n}$, we have:
\begin{align*}
\mathit{exp}\big( -(\alpha-q)^2/(2\alpha) \big)
&= \mathit{exp}\big( -(2oq/3 - q)^2/(4oq/3) \big)
\\& \approx \mathit{exp}( -\sqrt{n} ). 
\end{align*}
Consequently, every correct replica terminates with a probability of at least $1 - 2(n-f)\mathit{exp}\boldsymbol{(}-\Theta(\sqrt{n})\boldsymbol{)}$.
\end{proof}

\begin{theorem}[Theorem~\ref{thm:termination:infty}]
In \probft, every correct replica eventually decides a value with probability one.
\end{theorem}
\begin{proof}
Given a correct leader that proposes a value in a view $v$ after GST, every correct replica decides a value with high probability according to Theorem~\ref{cor:termination}.
However, a correct replica might not decide a value due to not receiving enough messages to form quorums.
Note that there will be an infinite number of views whose leaders are correct after view $v$.
As the number of views required to decide a value follows a geometric distribution~\cite{probBook} with parameter $p \approx 1 - 2(n-f)\mathit{exp}(-\sqrt{n})$, the probability of deciding a value by a correct replica in $k$ views whose leaders are correct is as follows:
\begin{align*}
    & \lim\limits_{k\rightarrow\infty} \Pr(\text{deciding a value in $k$ views with correct leaders})  \approx \lim\limits_{k\rightarrow\infty} 1 - (1-p)^k = 1.
\end{align*}   
\end{proof}

\subsection{Probability of Agreement within a View}\label{sec:safety_within_a_view}
In \probft, different replicas may decide different values since quorum intersections are not deterministic as in PBFT.
However, the protocol must ensure the probability of agreement violation is low.
To cause disagreement in a view of \probft, it is required that at least two distinct correct replicas decide two different values in that view.
Before addressing possible scenarios, it is essential to clarify that since the leader signs every proposal in \probft, a prerequisite for disagreement in a view is to have a faulty leader.

As discussed in Section~\ref{sec:correctness:proofs}, there are three scenarios in which agreement can be violated within a view.
The third scenario depicted in Figure~\ref{fig:3} is optimal.
Hence, we only consider this case in our analysis.

\begin{lemma}\label{lem:agreement:violation}
Given a Byzantine leader who may send at least two proposals~$\mathit{val}$ and~$\mathit{val}'$ and intends to maximize the probability of disagreement, the probability of forming a quorum by a correct replica for~$\mathit{val}$ and a quorum by another correct replica for~$\mathit{val}'$ is at most $\mathit{exp}\boldsymbol{(}-\Theta(\sqrt{n})\boldsymbol{)}^2$ if $r \leq n/o$.
\end{lemma}
\begin{proof}
According to our discussion in Section~\ref{sec:correctness:proofs}, the best strategy for the Byzantine leader is the one depicted in Figure~\ref{fig:3}.
Consequently, we can consider two sets $\Pi^1$ and $\Pi^2$, each with a size of $r = ((n-f)/2) + f$.
The leader sends a value~$\mathit{val}$ to~$\Pi^1$ and another value~$\mathit{val}'$ to~$\Pi^2$.  
Assume that~$i$ and~$j$ are two correct replicas such that~$i\in \Pi^1$ and~$j \in \Pi^2$.
We need to compute an upper bound for the probability that $i$ forms a probabilistic quorum for~$\mathit{val}$ while $j$ forms a probabilistic quorum for~$\mathit{val}'$.
With this aim, let $X$ be a random variable identifying the number of replicas within $\Pi^1$ that have~$i$ in their samples.
Note that $i$ can form a probabilistic quorum for~$\mathit{val}$ if $X \geq q$.
From Lemma~\ref{lem:expect}, $\mathop{\mathbb{E}}[X] = s\times r/n$.
As $s = o \times q$, $\mathop{\mathbb{E}}[X] = o\times q \times r/n$.
By assuming $\delta = (n/(o\times r)) -1$, when~$r \leq n/o$, we have:
\begin{align*}
    & \Pr(\text{$i$ forms a probabilistic quorum for~$\mathit{val}$})
    \\&\qquad \leq \  \Pr(X \geq q)
    \\&\qquad = \Pr\big(X \geq (1+\delta)o\times q\times r/n \big)
    \\&\qquad = \Pr\big(X \geq (1+\delta)\mathop{\mathbb{E}}[X] \big) 
    \\&\qquad \leq \mathit{exp}\big( -\delta^2\mathop{\mathbb{E}}[X]/(\delta+2) \big) \qquad (\text{using the Chernoff bound~\ref{ineq:chernof}})
    \\&\qquad = \mathit{exp}\big( -\delta^2\times o\times q\times r/(n(\delta+2)) \big). 
\end{align*}
Note that $r = \Theta(n)$, $\delta = \Theta(1)$, and $q = \Theta(\sqrt{n})$.
Hence, 
\begin{align*}
    \Pr(\text{$i$ forms a probabilistic quorum for~$\mathit{val}$})
    \leq \mathit{exp}\boldsymbol{(}-\Theta(\sqrt{n})\boldsymbol{)}.
\end{align*}
Note that forming a probabilistic quorum by~$i$ and forming a probabilistic quorum by~$j$ are negatively associated. 
Hence, we have:
\begin{align*}
    & \Pr(\text{$i$ and $j$ form probabilistic quorums})
    \\&\qquad \leq \Pr(\text{$i$ forms a probabilistic quorum}) \times 
    \Pr(\text{$j$ forms a probabilistic quorum})
    \\&\qquad \leq \mathit{exp}\boldsymbol{(}-\Theta(\sqrt{n})\boldsymbol{)}^2.
\end{align*}

\end{proof}

\begin{theorem}[Theorem~\ref{thm:agreement:in:view}]
Given a Byzantine leader who may send several proposals, the probability of agreement violation in a view is at most $\mathit{exp}\boldsymbol{(} -\Theta(\sqrt{n}) \boldsymbol{)}^4$. 
\end{theorem}
\begin{proof}
Given a Byzantine leader who sends at least two proposals~$\mathit{val}$ and~$\mathit{val}'$ and intends to maximize the probability of disagreement, the agreement can be violated if at least two replicas~$i$ and~$j$ decide different values.
Without loss of generality, assume that~$i$ decides~$\mathit{val}$ while~$j$ decides~$\mathit{val}'$.
Note that a replica can decide a value if it forms two consecutive quorums.
Let $A$ (resp. $B$) be an event that~$i$ forms a \textit{prepare} (resp. \textit{commit}) quorum for~$\mathit{val}$ and~$j$ forms a \textit{prepare} (resp. \textit{commit}) quorum for~$\mathit{val}'$.
According to Lemma~\ref{lem:agreement:violation}, the probability of happening~$A$ is bounded by $\mathit{exp}\boldsymbol{(}-\Theta(\sqrt{n})\boldsymbol{)}^2$.
Note that the probability of happening event~$B$ is less than the probability of happening event~$A$. 
Besides, note that event~$B$ might happen if event~$A$ has happened.
Hence, we have:
\begin{align*}
    \Pr(A \text{ and } B \text{ happen}) 
    \leq Pr(A)\times\Pr(B)
    \leq \mathit{exp}\boldsymbol{(}-\Theta(\sqrt{n})\boldsymbol{)}^4.
\end{align*}
\end{proof}

\subsection{Probability of Agreement with a View Change}\label{sec:safety_view_change}
\begin{lemma}\label{lem:d:3}
    The probability of deciding a value $\mathit{val}$ by a correct replica when $r \leq n/o$ replicas have prepared~$\mathit{val}$ is at most $\mathit{exp}\boldsymbol{(} -\delta^2\times o\times q\times r/(n(\delta+2)) \boldsymbol{)} $, where $\delta = (n/(o\times r)) -1$.
\end{lemma}
\begin{proof}
    Assume that a correct replica~$i$ decides a value $\mathit{val}$ when a set of replicas $R$ with a size of~$r$ have prepared $\mathit{val}$.
    Let $X$ be a random variable identifying the number of replicas within $R$ that have~$i$ in their commit samples.
    Note that $i$ can decide $\mathit{val}$ if $X \geq q$.
    From Lemma~\ref{lem:expect}, $\mathop{\mathbb{E}}[X] = s\times r/n$.
    As $s = o \times q$, $\mathop{\mathbb{E}}[X] = o\times q \times r/n$.
    By assuming $\delta = (n/(o\times r)) -1$, we have:
    \begin{align*}
        \Pr(\text{$i$ decides $\mathit{val}$})
        & \leq \  \Pr(X \geq q)
        \\&=\    \Pr\big(X \geq (1+\delta)o\times q\times r/n \big)
        \\&=\    \Pr\big(X \geq (1+\delta)\mathop{\mathbb{E}}[X] \big) 
        \\& \leq\ \mathit{exp}\big( -\delta^2\mathop{\mathbb{E}}[X]/(\delta+2) \big) \qquad (\text{using the Chernoff bound~\ref{ineq:chernof}})
        \\&=\    \mathit{exp}\big( -\delta^2\times o\times q\times r/(n(\delta+2)) \big). 
    \end{align*}
\end{proof}

\begin{theorem}[Theorem~\ref{thm:view:change}]
    The probability of proposing a value~$\mathit{val}'$ when another value $\mathit{val}$ has been decided by a correct replica in a prior view is at most $3 \times \mathit{exp}\boldsymbol{(}-(q\times \delta^2)/((\delta+1)\times(\delta+2)) \boldsymbol{)} $, where $\delta = 2n/(o\times (n+f)) -1$.
\end{theorem}
\begin{proof}
    Assume that a value~$\mathit{val}$ is decided by a correct replica in view~$v$.
    Also, assume that the leader~$l$ of view~$v+1$ is Byzantine.
    Recall that~$l$ needs to collect $\lceil (n+f+1)/2\rceil$ \textsc{NewView} messages to be allowed to propose a value.
    There are only three possible scenarios where~$l$ can propose another valid value~$\mathit{val}'$ in view~$v+1$:
    \begin{enumerate}
        \item Value~$\mathit{val}'$ was proposed in a view~$v' < v$, and there are at least $\lceil (n+f+1)/2\rceil$ replicas~$P$ such that each correct replica~$i \in P$ prepared~$\mathit{val}'$ in view~$v'$ but has not prepared any other value in a view after~$v'$.
        \item There are at least $\lceil (n+f+1)/2\rceil$ replicas that have prepared~$\mathit{val}'$ in view~$v$.
        \item There are at least $\lceil (n+f+1)/2\rceil$ replicas~$P$ such that each correct replica~$i \in P$ has not prepared any value in prior views; hence, the leader proposes its value. 
    \end{enumerate}
    
    Note that in each of these scenarios, value~$\mathit{val}$ must be prepared by at most~$(n-f)/2$ correct replicas.
    We now show that the probability of each scenario is negligible.
    Since~$f$ Byzantine replicas might have prepared~$\mathit{val}$, the number of replicas that have prepared~$\mathit{val}$ is less than~$((n-f)/2) + f = (n+f)/2$.
    From Lemma~\ref{lem:d:3}, the probability of deciding~$\mathit{val}$ when $r=(n+f)/2$ replicas have prepared~$\mathit{val}$ is at most $p=\mathit{exp}\boldsymbol{(}-(q\times \delta^2)/((\delta+1)\times(\delta+2)) \boldsymbol{)} $, where $\delta = 2n/(o\times (n+f)) -1$.
    Hence, the probability of each scenario is negligible and bounded by~$p$. 
    Since these scenarios are disjoint, the probability of proposing value~$\mathit{val}'$ is bounded by~$3p$.
\end{proof}

\begin{corollary}[Corollary~\ref{thm:livenss:safety}]
\probft guarantees safety with a probability of $1 - \mathit{exp}(-\Theta(\sqrt{n}))$.
\end{corollary}
\begin{proof}
    From Theorems~\ref{thm:agreement:in:view} and \ref{thm:view:change}, it follows that \probft guarantees safety with a probability of $1 - \mathit{exp}(-\Theta(\sqrt{n}))$.
\end{proof}
\end{document}

%% file: packages.tex
\usepackage{algorithmicx}
\usepackage{algorithm}
\usepackage[noend]{algcompatible}
\usepackage{subcaption}
\usepackage[inline]{enumitem}

\newcommand\NoDo{\renewcommand\algorithmicdo{}}

\newcommand\NoThen{\renewcommand\algorithmicthen{}}

\let\oldnl\nl
\newcommand{\nonl}{\renewcommand{\nl}{\let\nl\oldnl}}

\usepackage{xspace}
\newcommand{\probft}[0]{\textsc{ProBFT}\xspace}

\newcommand{\vrf}[0]{VRF\xspace}
\newcommand{\type}[1]{\textsc{#1}}
\newcommand{\msg}[3]{\ensuremath{\langle \type{#1}, #2 \rangle_{#3}}\xspace}

\theoremstyle{definition}
\newtheorem{definition}{Definition}
\newtheorem{theorem}{Theorem}
\newtheorem{lemma}{Lemma}
\newtheorem{corollary}{Corollary}

\newboolean{extendedVersion}
\setboolean{extendedVersion}{true}

\usepackage{adjustbox}